\documentclass[a4paper,USenglish]{lipics-v2019}
\usepackage{amsmath,amsthm}
\usepackage{color,soul}
\usepackage{graphicx}
\usepackage{hyperref}
\usepackage{cite}
\usepackage{tikz}
\usepackage[textsize=small, textwidth=2cm, color=cyan]{todonotes}
\usetikzlibrary{graphs,matrix}

\newcommand{\mb}[1]{\mathbf{#1}}
\newcommand{\bb}[1]{\mathbb{#1}}
\newcommand{\m}[1]{\mathcal{#1}}

\newcommand{\wf}{\textsf{\small WAIT-FREE}}
\newcommand{\tr}{\textsf{$f$-\small RESILIENT}}

\newcommand{\local}{\textsf{\small LOCAL}}
\newcommand{\hlocal}{\textsf{\small H-LOCAL}}
\newcommand{\wflocal}{\textsf{\small WF-LOCAL}}
\newcommand{\dyn}{\textsf{\small DYN}}
\newcommand{\hdyn}{\textsf{\small H-DYN}}
\newcommand{\lcl}{\textsf{\small LCL}}
\newcommand{\ld}{\textsf{\small LD}}
\newcommand{\name}{\operatorname{name}}
\newcommand{\lab}{\operatorname{label}}
\newcommand{\id}{\operatorname{id}}

\newcommand{\val}{\operatorname{val}}
\newcommand{\alg}{\mbox{{\sc alg}}}
\newcommand{\Cl}{\operatorname{Cl}}
\newcommand{\Sk}{\operatorname{Sk}}
\newcommand{\St}{\operatorname{St}}

\newcommand{\ora}[1]{\textcolor{black}{#1}}

\title{On Extending Brandt's Speedup Theorem from LOCAL to Round-Based Full-Information Models}

\titlerunning{On Extending Brandt's Speedup Theorem}

\author{Paul Bastide}{Ecole Normale Sup\'erieure de Rennes, France.}{paul.bastide@ens-rennes.fr}{}{}

\author{Pierre Fraigniaud}{Universit\'e de Paris and CNRS, France. Additional supports from ANR Projects DESCARTES and DUCAT.}{pierre.fraigniaud@irif.fr}{}{} 

\authorrunning{P. Bastide and P. Fraigniaud}
\Copyright{Paul Bastide and Pierre Fraigniaud}



\category{} 

\relatedversion{} 

\supplement{}

\acknowledgements{}

\begin{document}
\nolinenumbers
\maketitle

\begin{abstract}
Given any task~$\Pi$, Brandt's speedup theorem (PODC 2019) provides a mechanical way to design another task~$\Pi'$ on the same input-set as~$\Pi$ such that, for any $t\geq 1$, $\Pi$ is solvable in $t$~rounds if and only if $\Pi'$ is solvable in $t-1$ rounds. The theorem applies to the anonymous variant of the \local\/ model, in graphs with sufficiently large girth, and to locally checkable labeling (\lcl) tasks. In this paper, using combinatorial topology applied to distributed computing, we dissect the construction in Brandt's speedup theorem for expressing it in the broader framework of round-based models supporting full information protocols, which includes models as different as wait-free shared-memory computing with iterated immediate snapshots, and synchronous failure-free network computing. In particular, we provide general definitions for notions such as local checkability and local independence, in our broader framework. In this way, we are able to identify the hypotheses on the computing model, and on the tasks, that are sufficient for Brandt's speedup theorem to apply. More precisely, we identify which hypotheses are sufficient for the each direction of the if-and-only-if condition. Interestingly, these hypotheses are of different natures. Our general approach enables to extend Brandt's speedup theorem from \local\/ to directed networks, to hypergraphs, to dynamic networks, and even to graphs including short cyclic dependencies between processes (i.e., the large girth condition is, to some extend,  not necessary). The theorem can even be extended to shared-memory wait-free computing. In particular, we provide new impossibility proofs for consensus and perfect renaming in 2-process systems. 
\end{abstract}

\section{Introduction}


Given a complexity or computability result established for a distributed computing model~$\m{M}_1$, several questions can be raised. Does this result hold for another model~$\m{M}_2$? What makes this result true for~$\m{M}_1$ but not for~$\m{M}_2$, or what are the  features common to $\m{M}_1$ and $\m{M}_2$ that make the result true for both models?  For instance, if a result holds in the \local\/ model~\cite{Linial92,Peleg00}, is it because the model is synchronous? Is it because processes and communication links are failure-free? Is it because the network satisfies some property (e.g., large girth)? Is it because the problem satisfies some property (e.g., local checkability)?  

A typical example is Brandt's speedup theorem~\cite{Brandt19}.
This theorem essentially provides a mechanical way to construct a task~$\Pi'$ from any task~$\Pi$, on the same input set as~$\Pi$, such that, for every $t\geq 1$, $\Pi$ is solvable in $t$~rounds in \local\/ if and only if $\Pi'$ is solvable in $t-1$ rounds in \local. This theorem is an efficient tool for designing lower bounds. Indeed, starting from a task~$\Pi$, iterating the construction results in a series of tasks $\Pi^{(r)}, r\geq 1$, such that, for every $t\geq 1$, $\Pi$ is solvable in $t$~rounds if and only if $\Pi^{(r)}$ is solvable in $t-r$~rounds. In particular,  $\Pi^{(t)}$ is solvable in zero rounds, and demonstrating that $\Pi^{(t)}$ is actually not solvable in zero rounds establishes the lower bound $t+1$ for the round-complexity of~$\Pi$. 

Brandt's speedup theorem does not directly applies to \local, but to an anonymous variant of \local\/ on graphs with sufficiently large \emph{girth}. This is because the presence of identifiers assigned to the nodes prevents \emph{local-independence} to be satisfied, where the latter is a property that is essential for establishing the theorem. It is not trivial to formally express this property, but, roughly speaking, given the radius-$(t-1)$ views of two adjacent nodes~$v$ and $v'$ in some network~$G$, the presence of identifiers results in the fact that one cannot guarantee that two independent extensions of these two views into radius-$t$ views are compatible. Indeed, one extension may include a node~$w$ provided with the same identifier as a  node~$w'$ in the other extension, with $w\neq w'$, in contradiction with the fact that each identifier must be unique in the network. Local independence also imposes to consider graphs~$G$ with girth $g>2t-1$. Indeed, in graphs with girth $g\leq 2t-1$, two independent radius-$t$  extensions of the radius-$(t-1)$ views of $v$ and $v'$ may include a same node~$w$ provided with different identifiers, or with different inputs. This would result  into two non-compatible radius-$t$ extensions in the sense that there are no instances yielding the simultaneous presence of these two radius-$t$ views at two adjacent nodes.  

Also, Brandt's speedup theorem requires the tasks at hand to be \emph{locally checkable}. This property essentially says that, given an assignment of input-output values to the nodes, the correctness of the collection of output values with respect to the collection of input values can be established by merely inspecting the  values of each node and of its neighbors in the network.  In other words, a task is locally checkable if the correctness of an assignment of values to the nodes is defined as the conjunction of the local correctness of this assignment, where ``local'' refers to the closed neighborhood of each node\footnote{The notion of local checkability can be extended to neighborhood at distance~$k$ in a straightforward manner, for any fixed $k\geq 1$.}. Proper coloring and maximal independent set (MIS) are typical examples of locally checkable tasks in \local.

We can now rephrase our original questioning in the specific case of Brandt's speedup theorem: does this theorem holds in other models? For such a question to make sense, we restrict attention to models in which the notion of \emph{rounds} is defined, which naturally include synchronous models in networks with multiparty interactions, namely \emph{hypergraphs}, and synchronous models in networks that evolve with time, namely, \emph{dynamic networks}. Round-based models however include far more than just synchronous models in networks. For instance, asynchronous shared-memory computing with \emph{iterated immediate snapshots}, referred to as \wf\/ in the following, which is computationally equivalent to asynchronous read/write shared-memory computing with crash-prone processes, is round-based. The same holds for $t$-resilient computing, $0\leq t\leq n-1$, which is essentially the same as \wf, but where at most $t$~processes can crash. 

The \local\/ model has another feature. It supports \emph{full information} communication protocols. That is, whenever a process receives information from another process, one can assume that the latter has sent \emph{all} the data it acquired before the communication took place. This assumption enables the design of strong lower bounds, which hold even if the processes are not restricted in term of volume of communication. Also, the \local\/ model does not restrict the individual computational power of the processes. This assumption enables the design of unconditional lower bounds, which hold independently from  complexity or computability assumptions regarding the computing power of each individual process. All the models mentioned above support full-information protocols, and have unlimited individual computational power. 

So, making our questioning even more specific: Is there an analog of Brandt's speedup theorem for all  round-based models supporting full-information protocols with unlimited individual computational power? If not, what make the \local\/ model so special? If yes, for which models? Under which conditions? 

\subparagraph{Our Results.}

Using the framework provided by combinatorial topology applied to distributed computing, we give a general definition of speedup tasks for round-based models supporting full-information protocols (with unlimited individual computational power). Given a task $\Pi$ in the \local\/ model,  Brandt's speedup theorem constructs such a speedup task~$\Pi'=\Phi(\Pi)$. We then revisit Brandt's construction, that is, we dissect the nature of the operator~$\Phi$ transforming any task $\Pi$ into a task $\Pi'=\Phi(\Pi)$, for identifying the central assumptions allowing this construction to work in \local. They are two central assumptions: local checkability and local independence. We extend these two notions from the \local\/ model to round-based models supporting full-information protocols. We also extend Brandt's operator~$\Phi$ to all such models. We denote by~$\Phi^\star$ this extension. As a result, we are able to express a general speedup theorem, which roughly reads as follows. Let $\m{M}$ be a round-based model supporting full-information protocols, let $\Pi$ be a task, and let $t\geq 1$. The task $\Phi^\star(\Pi)$ satisfies the following:

\begin{enumerate}
    \item Assume that $\Pi$ satisfies $(t-1)$-independence with respect to~$\m{M}$. If $\Pi$ is solvable in at most $t$ rounds, then $\Phi^\star(\Pi)$ is solvable in at most $t-1$ rounds. 
    \item Assume that $\Pi$ is locally checkable in~$\m{M}$. If $\Phi^\star(\Pi)$ is solvable in at most $t-1$ rounds, then $\Pi$ is solvable in at most $t$ rounds. 
\end{enumerate}

Statement~1 guarantees that the task $\Phi^\star(\Pi)$ is at least ``1-round faster'' than the original task~$\Pi$. Statement~2 guarantees that $\Phi^\star(\Pi)$ is no more than ``1-round faster'', and in particular that $\Phi^\star(\Pi)$ is not solvable in zero rounds.  Observe that the sets of hypotheses required for each of the two statements are different, and actually they do not even intersect. This provides flexibility. For instance, given a task $\Pi$ satisfying local independence w.r.t. model~$\m{M}$, even if $\Pi$  is not locally checkable in~$\m{M}$, it may still be the case that, thanks to Statement~1, $\Phi^\star(\Pi)$ remains solvable in at least $f(t)$ rounds for some function~$f$, and that iterating $\Phi^\star$ results in a non-trivial lower bound. For instance, $f(t)=\log t$ is sufficient for deriving a lower bound $\Omega(\log^*n)$. Satisfying Statement~2 requires to limit the class of tasks under consideration to locally checkable tasks. For instance, proper-coloring and renaming are locally checkable in \local\/ and \wf, respectively, but spanning tree and consensus are not locally checkable in these respective models. 

Concretely, our general construction~$\Phi^\star$ allows us to directly extend Brandt's speedup theorem to various kinds of synchronous models in networks, including directed graphs, hypergraphs, dynamic networks, and even to graphs including short cyclic dependencies between processes (i.e., the large girth condition is, to some extend,  not necessary). Interestingly, our general construction also enables to extend Brandt's speedup theorem to asynchronous failure-prone computing models such as \wf. In particular, we provide a new impossibility proof for consensus and for perfect renaming in 2-process systems. 

\subparagraph{Related Work.} 


Roughly, the modern approach of distributed computing can be presented as the study of two large classes of computing models, one whose models are aiming at capturing issues related to \emph{time} (asynchrony, crashes, etc.)~\cite{AttiyaW04}, and another whose models are aiming at capturing issues related to \emph{space} (latency, congestion, etc.)~\cite{Peleg00}. The study of the former class puts emphasis on the study of system tasks such as leader election or consensus, while study of the latter class put emphasis on the study of graph problems such as coloring or matching. The applications of topology to the theory of distributed computing was introduced in~\cite{HerlihyS99,SaksZ00} for studying system tasks under asynchronous crash-prone computing models. This paper is inspired from~\cite{HerlihyS99}, which identified the properties of the topological deformations related to wait-free and $t$-resilient computing. Since then, topology has been extensively used in the context of asynchronous computing with crash-prone processes, for establishing lower bounds or impossibility results~\cite{AttiyaCHP19,CastanedaR10,FraigniaudRT20}, but also upper bounds~\cite{CastanedaR12}. It has also been extended to mobile computing~\cite{AlcantaraCFR19}, to dynamic environments~\cite{GodardP16}, and to Byzantine failures~\cite{MendesTH14}. Moreover, a topological description of concurrent programming has been developed~\cite{FajstrupGHMR16,GoubaultMT18}. It is however only recently that distributed network computing has been approached through the lens of combinatorial  topology~\cite{CastanedaFPRRT21,FraigniaudP20}, specifically applied to local computing. 

The \local\/ model is a synchronous failure-free model dedicated to capture local computing in networks, that is, the ability to solve problems by having each process inspecting solely the inputs present in its vicinity in the network (see~\cite{HirvonenS2020,Peleg00}). Among the earliest seminal work in the  \local\/ model are~\cite{Linial92} and~\cite{NaorS95}. The former established the celebrated lower bound $\Omega(\log^*n)$ rounds for 3-coloring the $n$-node cycle.  The latter introduced the class of \emph{locally checkable labeling} (\lcl) problems, that is, the class of problems defined on bounded-degree graphs, involving individual inputs and outputs of bounded size, and whose candidate solutions can be checked locally\footnote{The class of tasks that are locally checkable are sometimes referred to as ``the equivalent of NP''. This is however debatable, as the formal definitions of complexity classes such as NP typically involve \emph{proofs} provided by non-trustable oracles. In the context of distributed network computing,  the ``equivalents of NP'' may rather be the classes PLS, LCP, and NLD, respectively defined in~\cite{KormanKP10}, \cite{GoosS16}, and~\cite{FraigniaudKP13}.}. It was shown that it is undecidable whether a given \lcl\/ problem is solvable locally (i.e., in a constant number of rounds). It was also shown that if an \lcl\/ problem can be solved locally by a randomized algorithm, then it can be solved locally by a deterministic algorithm. This derandomization result initiated a vast literature on the power and limitation of randomized algorithms in their ability of solving problems locally (see, e.g., \cite{Balliu0OS20,ChangKP19,GhaffariKM17,RozhonG20} for recent contributions). The aforementioned reference~\cite{Linial92} introduced a lower bound technique bearing similarities with the topological approach, that connects a structural property of a graph capturing all possible configurations of the system at a given time~$t$ with the ability to solve a problem in $t$~rounds. For a quarter of a century, this was the only known non-trivial lower bound technique departing from using indistinguishably arguments, until the breakthrough~\cite{Brandt19} introducing the aforementioned speedup technique. This technique, designed for \lcl\/ problems in general, was successfully applied for deriving lower bounds on various problems such as sinkless orientation and 2-weak coloring~\cite{Brandt19}, as well as maximal matching and maximal independent set~\cite{Balliu0HORS19}. We refer to the recent paper~\cite{Suomela20} for more details on using the speedup technique to understand locality.  

The study of distributed algorithms for networks has recently been subject to generalizations from graphs to hypergraphs, for handling frameworks with multiparty interactions. In particular, the maximal independent set problem in hypergraphs was studied in~\cite{KuttenNP014}, and maximal matching in hypergraphs was studied in~\cite{FischerGK17}. Interestingly, that latter paper shows that solving problems on hypergraphs has also surprising implications on solving other problems efficiently on graphs. Various extensions of the maximal independent set problem were studied in hypergraphs in~\cite{KuhnZ2018}, specifically for linear hypergraphs (i.e., hypergraphs in which any two hyperedges overlap on at most one node). We refer to this latter paper for pointers on earlier contributions on the design of distributed algorithms for hypergraphs. 
 
\section{Summary of our Contributions}
\label{sec:summary}

This section is a technical  summary of our approach and main results. The  model considered in this section is not the most general one, and our general model will be introduced further in the paper. In particular, the model presented in this section does not capture hypergraphs. Nevertheless, it is sufficient for presenting our main ideas and techniques. 

\subsection{Distributed Computation} 

Combinatorial topology provides an elegant and unified way to describe distributed computing (see~\cite{Herlihy2013,HerlihyS99}). We refer to Section~\ref{sec:general-framework} for the details, but, roughly, a \emph{task} (e.g., consensus, vertex-coloring, renaming, maximal independent set, etc.) for an $n$-process system with processes $p_1,\dots,p_n$ can be defined as a triple $(\m{I},\m{O},\Delta)$, where $\m{I}$ and $\m{O}$ respectively denote the sets of all legal $k$-process input and output states, $1\leq k\leq n$, and $\Delta:\m{I}\to 2^{\m{O}}$ is a function that maps every input state $\sigma\in\m{I}$ to the set of  output states $\Delta(\sigma)$ that are legal w.r.t.~$\sigma$. See Fig.~\ref{fig:topoiff}.

\begin{figure}[tb]
\begin{center}
\begin{tikzpicture}
\small
\matrix (m) 
[matrix of math nodes,row sep=3em,column sep=4em,minimum width=2em] 
{
 \m{I} & \m{P}^{(t)}  \\
 & \m{O}     \\ 
}; 
\path[-stealth] 
(m-1-1) edge node [above] {$\Xi^t$} (m-1-2)
(m-1-1) edge node [above] {$\Delta$} (m-2-2)
(m-1-2)  edge node [right] {$\delta$} (m-2-2);
\end{tikzpicture}
\end{center}
\vspace*{-3ex}
\caption{\sl The topological approach of distributed computing}
\label{fig:topoiff}
\end{figure}
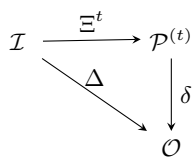

Formally, $\m{I}$ and $\m{O}$ are \emph{simplicial complexes}, and every set $\sigma$ in one of these two complexes is a \emph{simplex} (see Appendix~\ref{app:element-topo} for a brief introduction to combinatorial topology). A simplex $\sigma\in\m{I}$ (resp., $\sigma\in\m{O}$) is of the form $\sigma=\{(i,x_i):i\in I\}$, where (1)~$I\subseteq [n]$ is any non-empty set of process \emph{names}, (2)~for every $i\in I$, $x_i$ is an input (resp., output) value, and (3)~these values are mutually compatible whenever, for every $i\in I$, $x_i$ is assigned to~$p_i$. A 0-dimensional simplex $\{(i,x_i)\}$, for some $i\in[n]$, is called a \emph{vertex}.  It is assumed that $\Delta$ preserves names, that is, if $\tau\in\Delta(\sigma)$, then $\name(\tau)=\name(\sigma)=I$. We denote by $\val(\m{I})$ and $\val(\m{O})$ the set of input and output values, respectively.  That is, for $\sigma=\{(i,x_i):i\in I\}$, if $\sigma\in\m{I}$ (resp., $\sigma\in\m{O}$) then $x_i\in\val(\m{I})$ (resp., $x_i\in\val(\m{O})$) for every $i\in I$.

We consider any \emph{round-based} communication model~$\m{M}$ supporting \emph{full information} protocols, which include, e.g., wait-free computing with iterated immediate snapshots, referred to as \wf~\cite{AttiyaW04,Herlihy2013}, in the context of shared-memory computing, and \local\/ in the context of network computing~\cite{Peleg00}. A crucial feature shared by all these models is that, w.l.o.g., one can restrict attention to algorithms decomposed into two phases: one phase consisting of a certain number~$t$ of  communication rounds where, at each round, each process forwards all the information acquired during the previous rounds, and one phase of computation in which an output is computed based on all the information accumulated during the $t$ communication rounds performed during the first phase. So, in fact, designing a $t$-round algorithm boils down to designing an output function mapping views gathered within $t$ communication rounds to output values. 

Let $\m{P}^{(t)}$ denote the set of all possible $k$-process states of the system, $1\leq k\leq n$, after $t$~rounds, with $\m{P}^{(0)}=\m{I}$. Like $\m{I}$ and $\m{O}$, $\m{P}^{(t)}$ is a simplicial complex, for every $t\geq 0$. For $t>0$, the complex $\m{P}^{(t)}$ is the image of $\m{P}^{(t-1)}$ by a function $\Xi$, which is specific of the communication model~$\m{M}$, and which is mapping every state $\sigma\in \m{P}^{(t-1)}$ to the set $\Xi(\sigma)\subseteq \m{P}^{(t-1)}$ of states that may result from $\sigma$ after one round of communication. As for the input-output specification~$\Delta$, if $\tau\in\Xi(\sigma)$ then  $\name(\tau)=\name(\sigma)$. Note that, in particular, $\m{P}^{(t)}=\Xi^t(\m{I})$, as displayed on Fig.~\ref{fig:topoiff}. 
For instance, given $\sigma=\{(i,x_i):i\in I\}\in\m{P}^{(t-1)}$, we have: 
\begin{itemize}
\item In \wf, $x_i$ is the view of $p_i$ resulting from its $(t-1)$th snapshots. We have $\tau\in \Xi(\sigma)$ if $\tau=\{(i,\{x_j:j\in J_i\}):i\in I\}$ where, for every $i\in I$,  (1)~$J_i\subseteq I$, (2)~$J_i\subseteq J_j$ or $J_i \supseteq J_j$ for every $j\in I$, and (3)~for every $j\in I$, if $j\in J_i$ then $J_j\subseteq J_i$.  
\item In \local, $x_i$ is the labeled ball of radius $t-1$ centered at $p_i$ in the input graph~$G$. Assuming $I=[n]$, $\tau\in \Xi(\sigma)$ if $\tau=\{(i,\{x_j:j\in N_G[i]\}):i\in [n]\}$ where $N_G[i]$ denotes the closed neighborhood of node~$i$ in~the underlying network~$G$. 
\end{itemize}
More generally, we model a communication model $\m{M}$ as a simplicial complex whose simplices are of the form  $\varphi=\{(i,J_i):i\in I\}$ with $i\in J_i\subseteq [n]$ for every $i\in I$. Such a simplex corresponds to a possible communication round in which, for every $i\in I$, process~$i$ receives information from all processes~$j\in J_i$. See Figs.~\ref{fig:exampleIOP} and~\ref{fig:examplerenaming} for examples in \local\/ and \wf. A simplex $\varphi=\{(i,J_i):i\in I\}$ of~$\m{M}$ is said to be \emph{closed} if $\cup_{i\in I}J_i=I$. To every communication model $\m{M}$ corresponds a communication map~$\Xi$. Let $\sigma=\{(i,v_i):i\in I\} \in \m{P}^{(t)}$ for some $t\geq 0$, where $I\subseteq [n]$, and let us assume that there exists a closed simplex $\varphi=\{(i,J_i):i\in I\}$ in~$\m{M}$. We set 
$
\Xi(\sigma,\varphi)=\big \{\big (i,\{v_j:j\in J_i\}\big ):i\in I\big \}, 
$
and we define
$
\Xi(\sigma)=\{\Xi(\sigma,\varphi) : (\varphi \in \m{M})\wedge(\name(\varphi)=\name(\sigma))\wedge(\mbox{$\varphi$ is closed})\}.
$

\begin{figure}[tb]
\centering
\includegraphics[width=14cm]{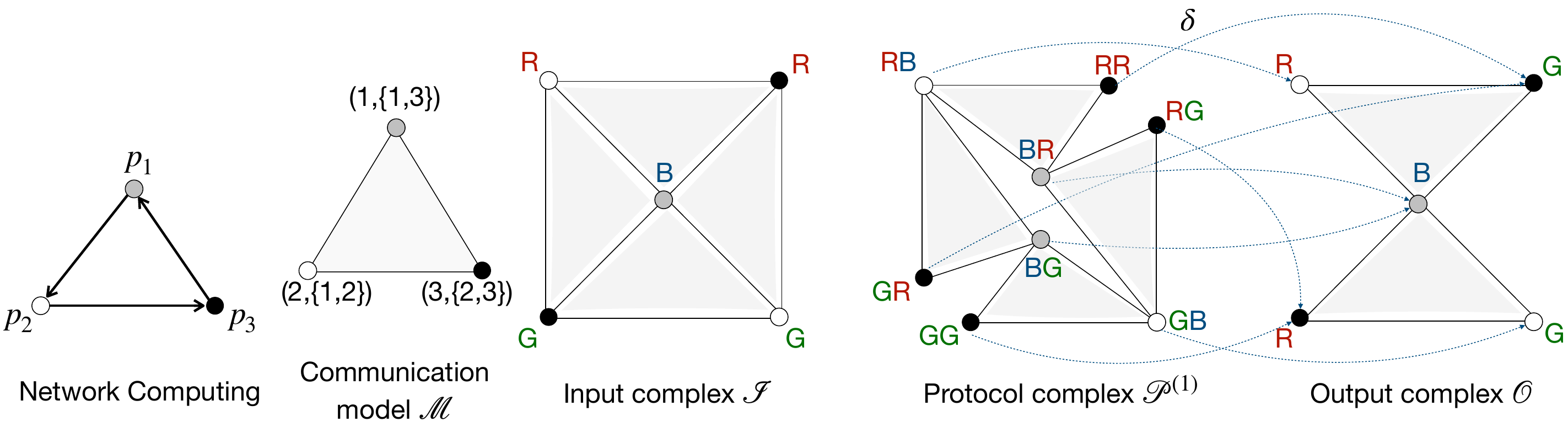}
\caption{\sl The communication model is the (synchronous failure-free) directed cycle $\vec{C}_3$. In the considered task, $p_1$ receives input $B$ (blue) as input, while $p_2$ and $p_3$ may receive input~R (red) or~G (green). The output complex specifies that the output values must be a proper 3-coloring of $\vec{C}_3$, with $p_1$ colored blue. After one round, the state of each process is a pair $XY$ of colors, where $X$ is its input color, and $Y$ is a color received from the in-neighbor in $\vec{C}_3$, forming the protocol complex $\m{P}^{(1)}$. As an input-output specification~$\Delta$, we consider the case where $p_2$ must output the same color as its input color. The map $\delta:\m{P}^{(1)}\to\m{O}$ depicted on the figure is simplicial and agrees with~$\Delta$. Note that $\delta$ is also name-independent. There are no simplicial maps from $\m{I}=\m{P}^{(0)}$ to~$\m{O}$ that agree with~$\Delta$, and therefore the task is not solvable in zero rounds (even if one discards name-independence).  
}
\label{fig:exampleIOP}
\end{figure}

A $t$-round algorithm is then a function $\delta:\m{P}^{(t)}\to \m{O}$ mapping every pair~$(i,x_i)\in \m{P}^{(t)}$ to some pair~$(i,y_i)=\delta(i,x_i)\in\m{O}$ (cf. Fig.~\ref{fig:topoiff}). Note that $\delta$ is \emph{name-preserving}. The semantic of this map is that process $i$ in state~$x_i$ outputs~$y_i$. In wait-free computing, $\delta$ essentially takes views resulting from $t$~rounds of iterated immediate snapshots as inputs, while, in \local, $\delta$ takes  labeled balls of radius~$t$ as inputs. The function~$\delta$ must satisfy two constraints: 
\begin{itemize}
\item $\delta$ is \emph{simplicial}, that is, for every $\sigma=\{(i,x_i):i\in I\}\in\m{P}^{(t)}$, $\delta(\sigma)=\{\delta(i,x_i):i\in I\}\in\m{O}$, i.e., $\delta(\sigma)$ is a legal $k$-process output state, where $k=|I|$, and
\item $\delta$ \emph{agrees} with~$\Delta$, that is, given any input state $\sigma\in \m{I}$, $\delta(\Xi^t(\sigma))\subseteq \Delta(\sigma)$, i.e., the output  after $t$ rounds  of a set of processes initially in state $\sigma\in\m{I}$ must be one of the output states  that are legal w.r.t.~$\sigma$. 
\end{itemize}
This formalism yields a characterization of task solvability (cf. Fig.~\ref{fig:topoiff}). 

\begin{lemma}
A task $(\m{I},\m{O},\Delta)$ is  solvable in at most $t$ rounds if and only if there exists a simplicial map $\delta:\m{P}^{(t)}\to \m{O}$ that agrees with~$\Delta$. 
\end{lemma}

Depending on the context, one usually requires that $\delta$ is \emph{name-independent}, that is, if $\delta(i,x)=(i,y)$ and $\delta(j,x)=(j,y')$, then $y=y'$, reflecting the fact that the name of a process is external, and not part of its input.  See Figs.~\ref{fig:exampleIOP} and~\ref{fig:examplerenaming} for examples. 

\begin{figure}[tb]
\centering
\includegraphics[width=14cm]{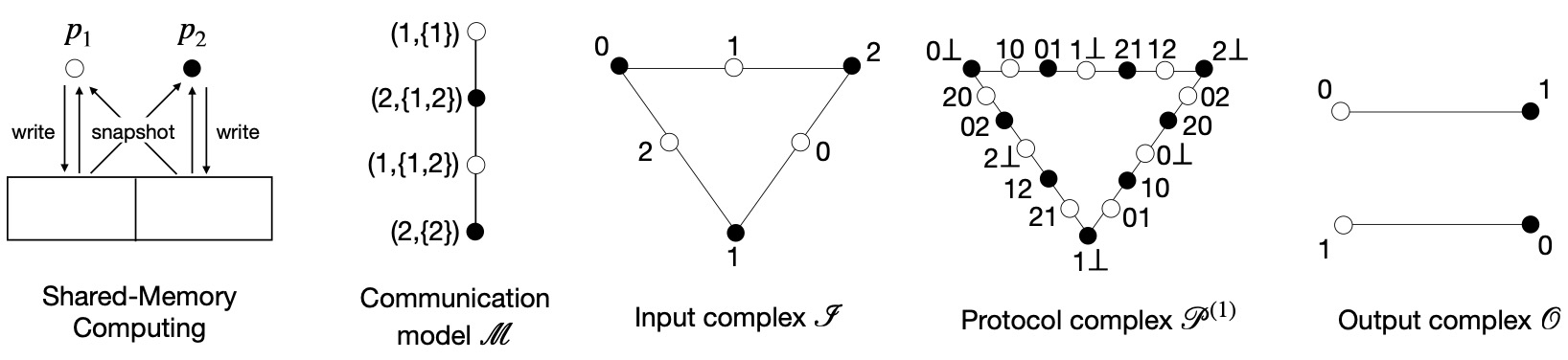}
\caption{\sl Perfect renaming in \wf: Each of the two processes receives an identifier $\id$ in $\{0,1,2\}$ such that $\id(p_1)\neq\id(p_2)$, and they must find new identifiers $\id'$ in $\{0,1\}$, respecting $\id'(p_1)\neq\id'(p_2)$. The protocol complex $\m{P}^{(1)}$ is a chromatic  subdivision of~$\m{I}$~\cite{HerlihyS99}. Perfect renaming is not solvable in 1~rounds, because there is no name-independent simplicial map from $\m{P}^{(1)}$ to~$\m{O}$. The same holds for any $t\geq 1$, and therefore perfect renaming is impossible in \wf. In this paper, we provide another proof of this result, using a generalized version of Brandt's speedup theorem. 
}
\label{fig:examplerenaming}
\end{figure}

\subsection{Speedup Tasks} 

In this section, we introduce a general notion of \emph{speedup tasks}. See Section~\ref{sec:speed} for more details. 

\subparagraph{Definition.} 

The speedup of a task $(\m{I},\m{O},\Delta)$ for a (full information) communication model~$\m{M}$ is a task $(\m{I},\m{O}',\Delta')$ such that (see Fig.~\ref{fig:speedup}):
\begin{itemize}
    \item for every $t\geq 1$, if there exists a simplicial map $\delta:\m{P}^{(t)}\to \m{O}$ that agrees with~$\Delta$, then there exists a simplicial map $\alpha:\m{P}^{(t-1)}\to \m{O}'$ that agrees with $\Delta'$, and  
    \item there exists a simplicial map $\beta:\Xi(\m{O}')\to \m{O}$ such that, for every $\sigma\in \m{I}$, ${\beta(\Xi(\Delta'(\sigma)))\subseteq \Delta(\sigma)}$.
\end{itemize}
Note that $\Xi(\m{O}')$ is the set of all global states that may result after a single round of communication under~$\m{M}$, starting from input states in~$\m{O}'$. The terminology ``speedup task'' is motivated by the following simple observation. 

\begin{lemma}
For every $t\geq 1$, there is a $t$-round algorithm for $(\m{I},\m{O},\Delta)$  in model~$\m{M}$ if and only if there exists a $(t-1)$-round algorithm for its speed up task $(\m{I},\m{O}',\Delta')$   in model~$\m{M}$. 
\end{lemma}

Indeed, the simplicial map $\alpha$ guarantees the solvability of $(\m{I},\m{O}',\Delta')$ in $t-1$ rounds assuming that $(\m{I},\m{O},\Delta)$ is solvable in $t$ rounds, and the simplicial map $\beta$ guarantees the solvability of $(\m{I},\m{O},\Delta)$ in $t$ rounds assuming that $(\m{I},\m{O}',\Delta')$ is solvable in $t-1$ rounds.

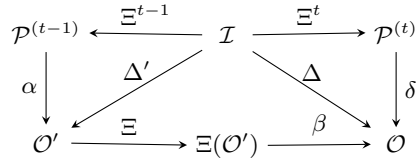
\begin{figure}[b]
\begin{center}
\begin{tikzpicture}
\small
\matrix (m) 
[matrix of math nodes,row sep=3em,column sep=4em,minimum width=2em] 
{
\m{P}^{(t-1)} & \m{I} & \m{P}^{(t)}  \\
 \m{O}' & \Xi(\m{O}') & \m{O}     \\ 
}; 
\path[-stealth] 
(m-1-2) edge node [above] {$\Xi^t$} (m-1-3)
(m-1-2) edge node [above] {$\Delta$} (m-2-3)
(m-1-3)  edge node [right] {$\delta$} (m-2-3)
(m-1-2) edge node [above] {$\Xi^{t-1}$} (m-1-1)
(m-1-1)  edge node [left] {$\alpha$} (m-2-1)
(m-2-1) edge node [above] {$\Xi$} (m-2-2)
(m-2-2) edge node [above] {$\beta$} (m-2-3)
(m-1-2) edge node [above] {$\Delta'$} (m-2-1)
;
\end{tikzpicture}
\end{center}
\vspace*{-3ex}
\caption{\sl The task $(\m{I},\m{O}',\Delta')$ is a speedup task for $(\m{I},\m{O},\Delta)$}
\label{fig:speedup}
\end{figure}

\subparagraph{Generic Approach for Constructing Speedup Tasks.}

There is a generic approach for constructing speedup tasks. To see how, note that, given its state $v_i$ after $t-1$ rounds, every process~$p_i$ can  internally build all its possible futures after one more round, as well as all the possible futures of all the other processes, whenever the current states of the other processes after  $t-1$ rounds are compatible with~$p_i$ in state~$v_i$.  Let $\St(i,v_i)$ denote the \emph{star} of $(i,v_i)$ in $\m{P}^{(t-1)}$, that is, the set of simplices of $\m{P}^{(t-1)}$ containing $(i,v_i)$ (see Fig.~\ref{fig:exampleClSt}(a)). Let $\Cl(\St(i,v_i))$ be the \emph{closure} of this star, that is, the minimal complex containing all simplices in~$\St(i,v_i)$. 

\begin{figure}[tb]
\centering
\includegraphics[width=10cm]{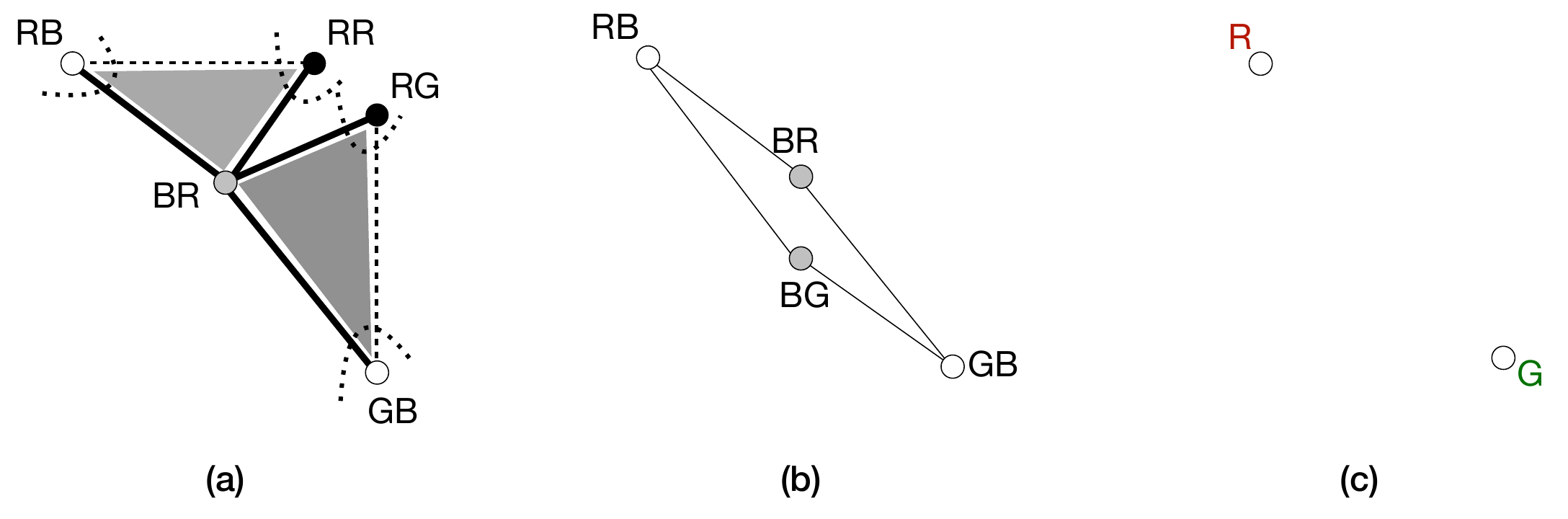}
\caption{\sl (a) The star of the vertex $(1,BR)$ in the protocol complex $\m{P}^{(1)}$ depicted on Fig.~\ref{fig:exampleIOP} consists of the vertex $(1,BR)$ itself, the four bold edges, and the two dark gray triangles. The closure of this star is the complex obtained by adding the four vertices $(2,RB)$, $(3,RR)$, $(3,RG)$, and $(2,GB)$, plus the two edges $\{(2,RB),(3,RR)\}$ and $\{(2,GB),(3,RG)\}$. (b)~The skeleton $\Sk_{\{1,2\}}(\m{P}^{(1)})$. (c)~The skeleton $\Sk_{\{2\}}(\m{O})$. 
}
\label{fig:exampleClSt}
\end{figure}

The complex $\Cl(\St(i,v_i))$ precisely captures all the possible states of the system after round $t-1$, given that process~$i$ is in state~$v_i$. Process~$i$ in state~$v_i$ can then compute $\Xi(\Cl(\St(i,v_i)))$, by simulating all possible scenarios resulting from one more round of communication. It follows that, after round $t-1$, process~$i$ in state~$v_i$ can output 
\[
\alpha(i,v_i)=\delta\Big(\Xi\big(\Cl(\St(i,v_i))\big)\Big).
\]
Observe that, for every $(i,v_i)\in \m{P}^{(t-1)}$, $\alpha(i,v_i)$ is a subcomplex of the output complex~$\m{O}$ of the task $(\m{I},\m{O},\Delta)$ at hand. This provides us with the intuition that, if the task $(\m{I},\m{O},\Delta)$ has a speedup task $(\m{I},\m{O}',\Delta')$, the complex~$\m{O}'$, as well as the input-output specification~$\Delta'$ are in close connection with the subcomplexes of~$\m{O}$, with the idea in mind that every process can extrapolate its current state by simulating all possible scenarios resulting from one more round of communication. In short, we foresee that a vertex of $\m{O}'$ is of the form $(i,K)$ where $i\in[n]$, and $K\subseteq \m{O}$ is a complex. Since a simplicial complex is a collection of sets of values, this provides us with the intuition that an output value for the speedup task $(\m{I},\m{O}',\Delta')$ is a set of set of output values in~$\m{O}$. Now, the next question to address is, what are the consistency conditions to be satisfied for a set $\{(i,K_i):i\in I\}$ of vertices of~$\m{O}'$ to be a simplex of~$\m{O}'$? At this stage of the discussion, it is not yet clear what these conditions should look like, but one can identify one hypothesis that helps very much, called \emph{local independence}. Roughly, given a simplex $\{(i,v_i),(j,v_j)\}\in\m{P}^{(t-1)}$, we would like all simplices in $\alpha(i,v_i)=\delta\big(\Xi\big(\Cl(\St(i,v_i))\big)\big)$ and $\alpha(j,v_j)=\delta\big(\Xi\big(\Cl(\St(j,v_j))\big)\big)$ to be compatible. 

\subparagraph{The Local Independence Property.}

Let $t\geq 1$ be an integer. A task $(\m{I},\m{O},\Delta)$ satisfies the \emph{$t$-independence} property w.r.t.~a communication model $\m{M}$ if, for every closed simplex $\varphi=\{(i,J_i) : i \in I \} \in \m{M}$, for every $i \in I$, for every $j \in J_i$, for every simplex ${\{(i,v_i),(j,v_j)\}\in \m{P}^{(t)}}$, and for every two collections of dimension-1 simplices
\[
\left\{\begin{array}{l}
C_i ={\big\{\sigma_{k} = \{ (i,v_i),(k,v_k) \} \in \St(i,v_i) : k \in J_i \smallsetminus \{i,j\}\big\}}\\
C_j ={\big\{\sigma_{k} = \{ (j,v_j),(k,v_k) \} \in \St(j,v_j) : k \in J_j \smallsetminus \{i,j\}\big\}}
\end{array}\right.
\]
we have 
$
\bigcup_{\sigma\in C_i\cup C_j} \sigma  \in \m{P}^{(t)}. 
$

Note that the $t$-independence property depends solely of the model~$\m{M}$, and of the input complex~$\m{I}$ of the task. Indeed, $\m{I}$ and~$\m{M}$ are the only parameters that govern the properties of the protocol complex at time~$t$. An interesting special case of the $t$-independence property is when considering  $j=i$. The $t$-independence property then implies that, for every $(i,J_i) \in \m{M}$, for every vertex $(i,v_i)\in \m{P}^{(t)}$, and for every collection $C ={\big\{\sigma_{k} = \{ (i,v_i),(k,v_k) \} \in \St(i,v_i) : k \in J_i \smallsetminus \{i\}\big\}}$, we have 
$
 \bigcup_{\sigma\in C} \sigma \in \m{P}^{(t)}.
$

Note that any task $(\m{I},\m{O},\Delta)$ for two processes satisfies the $t$-independence property w.r.t.~any communication model $\m{M}$, for any $t\geq 0$. The task depicted on Fig.~\ref{fig:exampleIOP}  satisfies 0-independence w.r.t.~$\vec{C}_3$ because, for every $\{(i,v_i),(j,v_j)\}\in \m{P}^{(0)}=\m{I}$, either $J_i=\{i,j\}$ or $J_j=\{i,j\}$, and any input value~$v_k$ of the third process~$p_k$ is compatible with $v_i$ and~$v_j$, i.e., $\{(i,v_i),(j,v_j),(k,v_k)\}\in \m{P}^{(0)}$. On the other hand, it is not 1-independent. To see why, let us consider the simplex $\{(1,BR),(2,RB)\}\in \m{P}^{(1)}$. We have $\{(1,BR),(3,RG)\}\in \St(1,BR)$, but $\{(1,BR),(2,RB),(3,RG)\}\notin \m{P}^{(1)}$. Yet, every task with locally checkable inputs satisfies the $t$-independence property w.r.t.~the anonymous variant of \local\/ in graph with girth larger than $2t+1$.  Indeed, in \local\/ the values $v_i$ in $\m{P}^{(t)}$ are input-labeled balls of radius~$t$ centered at~$p_i$, and $t$-independence boils down to the ability to extend these balls into balls of radius $t+1$ in a compatible manner for any two adjacent processes. See~\cite{Brandt19} for more details. 

\subsection{A General Construction of Speedup Tasks}

Let $(\m{I},\m{O},\Delta)$ be a task for $n$ processes. For defining a speedup task $(\m{I},\m{O}',\Delta')$, 
we define the complex~$\m{O}'$, and the input-output specification~$\Delta'$, as follows.  The construction is inspired by the aforementioned map $\alpha:\m{P}^{(t-1)}\to \m{O}'$ defined by $\alpha(i,v_i)=\delta\big(\Xi\big(\Cl(\St(i,v_i))\big)\big)$, but every complex resulting from the application of~$\alpha$ at a vertex $(i,v_i)$ of $\m{P}^{(t-1)}$ is decomposed into sets of vertices. More specifically, let $I\subseteq [n]$. We denote by $\Sk_I(\m{O})$ the \emph{skeleton}  of~$\m{O}$ composed of all simplices $\sigma\in\m{O}$ with $\name(\sigma)\subseteq I$ (see Fig.~\ref{fig:exampleClSt}(b)). In particular,  $\Sk_{\{i\}}(\m{O})$ is merely a set of vertices of~$\m{O}$, each of the form $(i,y)$ for some $y\in \val(\m{O})$ (see Fig.~\ref{fig:exampleClSt}(c)). First, we describe the vertices of~$\m{O}'$, then its simplices, and finally the input-output specification~$\Delta'$. See Section~\ref{sec:generalizedBrandt} for more details. 

\subparagraph{Vertices of~$\m{O}'$.} 

Each vertex in~$\m{O}'$ is a pair $(i,\bb{P}_i)$ with 
$\bb{P}_i = (x_i , \bb{S}_i)$,  
$x_i\in \val(\m{I})$, and 
\[
\bb{S}_i=\{\bb{S}_{i,k,J_i} : (i,J_i) \in \m{M} \;\mbox{and}\; k \in  [n] \},
\]
where,  for every vertex $(i,J_i) \in \m{M}$, and for every $k\in [n]$, 
$
\bb{S}_{i,k,J_i} \in 2^{2^{\Sk_{\{i\}}(\m{O})}}.
$
In other words, each $\bb{S}_{i,k,J_i}$ is a collection of sets with elements in $\Sk_{\{i\}}(\m{O})$. There is a set $\bb{S}_{i,k,J_i}$ for every set $J_i$ of processes from which process~$i$ may receive information in some communication specified by~$\m{M}$, and for every process $k\in [n]$. 
Each set $\bb{S}_{i,k,J_i}\in \bb{S}_i$ is identified in $ \bb{S}_i$ by a pair $(k,J_i)$. That is, formally, $\bb{S}_i$~is  an array indexed by pairs (process, set of processes). Nevertheless, for the sake of simplifying the notations, we describe $\bb{S}_i$ as a set. For a pair $(i,(x_i , \bb{S}_i))$ to be a vertex of~$\m{O}'$, the sets $\bb{S}_{i,k,J_i}$ in $\bb{S}_i$ must satisfy the following property:   

\begin{description}
\item[P0:] For every vertex $(i,J_i) \in \m{M}$ with $J_i = \{k_1,\ldots,k_d\}$, and for every 
\[
(S_{i,k_1,J_i},\ldots,S_{i,k_d,J_i}) \in \bb{S}_{i,k_1,J_i} \times \ldots \times \bb{S}_{i,k_d,J_i},
\]
we have 
$
\bigcap\limits_{j \in [d]} S_{i,k_j,J_i} \neq \varnothing.
$
\end{description}

\subparagraph{Example.} 

Let us for instance consider the task of Fig.~\ref{fig:exampleIOP}. We have $J_1=\{1,3\}$, and thus there are three sets in $\bb{S}_1$, which are $\bb{S}_{1,1,\{1,3\}},\bb{S}_{1,2,\{1,3\}},\bb{S}_{1,3,\{1,3\}}$. We have $\Sk_{\{1\}}(\m{O})=\{(1,B)\}$, and thus each of these three sets is potentially one of the four  sets of sets with elements in $\Sk_{\{1\}}(\m{O})$, namely $\varnothing, \{\varnothing\}, \{\{B\}\}, \{\varnothing,\{B\}\}$, where, for the sake of simplifying the notations, we denote by $B$ the vertex $(1,B)\in \Sk_{\{1\}}(\m{O})$. However, the sets $\{\varnothing\}$ and $\{\varnothing,\{B\}\}$ do not satisfy P0, and therefore only the two sets $\varnothing$ and $\{\{B\}\}$ remain. For $i\in\{2,3\}$, we have $\Sk_{\{i\}}(\m{O})=\{(i,R),(i,G)\}$ (cf. Fig.~\ref{fig:exampleClSt}). Therefore, still denoting by $X$ a vertex $(i,X)\in \Sk_{\{i\}}(\m{O})$, we get a larger collection of sets, including, e.g., $\{\{R\},\{G\}\}$ and $\{\{R,G\}\}$. Note however, that $\{\{R\},\{G\}\}$ can occur at most once in $\bb{S}_i$ because it is not true that for any $S\in \{\{R\},\{G\}\}$, and any $S'\in \{\{R\},\{G\}\}$, we have $S\cap S'\neq \varnothing$. For instance, $\{R\}\cap\{G\}=\varnothing$.  

\subparagraph{Simplices of~$\m{O}'$.} 

A vertex-set $\{(i,(x_i,\bb{S}_i)): i \in [n]\}$ is a facet of $\m{O}'$ if, for every closed simplex $\{(i,J_i) : i \in I \} \in \m{M}$, and for every $(i,k) \in I\times I$ with $i \in J_k$ or $k \in J_i$, the following two properties hold:

\begin{description}
\item[P1:] There exists $(S_{i,k,J_i,J_k},S_{k,i,J_k,J_i}) \in \bb{S}_{i,k,J_i} \times \bb{S}_{k,i,J_k}$ satisfying that, for every 
\[ 
\big ((i,y_i),(k,y_k)\big) \in S_{i,k,J_i,J_k} \times S_{k,i,J_k,J_i},
\]
there exists $\tau \in \m{I}$ such that
$
\{(i,x_i), (k,x_k)\} \subseteq \tau,
\;\mbox{and}\;
 \{ (i,y_i),(k,y_k) \} \in \Cl(\Delta(\tau)) . 
$

\medbreak 

\noindent Moreover, for every $(k,J_k') \in \Cl(\St(i,J_i))$, and for every $(i,J_i') \in \Cl(\St(k,J_k))$,
\[
S_{i,k,J_i,J_k} = S_{i,k,J_i,J_k'} \;\mbox{and} \; S_{k,i,J_k,J_i} = S_{k,i,J_k,J_i'}.
\]

\item[P2:] $|\bb{S}_{i,i,J_i}| = 1$ (i.e., a unique set in~$\bb{S}_{i,i,J_i}$), and, if $i \in J_k$ and $k \notin J_i$ then $\bb{S}_{i,k,J_i} = \bb{S}_{i,i,J_i}$.
\end{description}

\subparagraph{Example.} 

Let us consider again the task of Fig.~\ref{fig:exampleIOP}, and let us define the following sets $\bb{S}_{1},\bb{S}_{2}$, and $\bb{S}_{3}$, where $X\in\{R,G\}$: 
\[
\begin{array}{lll}
\bb{S}_{1,1,\{1,3\}}=\{\{B\}\} & \bb{S}_{1,2,\{1,3\}}=\{\{B\}\} & \bb{S}_{1,3,\{1,3\}}=\{\{B\}\} \\
\bb{S}_{2,1,\{1,2\}}=\{\{X\}\} & \bb{S}_{2,2,\{1,2\}}=\{\{X\}\} & \bb{S}_{2,3,\{1,2\}}=\{\{X\}\} \\
\bb{S}_{3,1,\{2,3\}}=\{\{R,G\}\} & \bb{S}_{3,2,\{2,3\}}=\{\{R\},\{G\}\} & \bb{S}_{3,3,\{2,3\}}=\{\{R,G\}\} 
\end{array}
\]
We claim that $\{(i,(x_i,\bb{S}_i)): i \in \{1,2,3\}\}$, where $x_1=B$, $x_2=X$, and $x_3\in \{R,G\}$, is a facet of $\m{O}'$. First, P0 is satisfied at each vertex $(i,(x_i,\bb{S}_i)$, $i\in\{1,2,3\}$. Second, for every $i\in\{1,2,3\}$, the set $\bb{S}_i$ satisfies~P2. It remains to check~P1. The only non-trivial case is $i=2$ and $k=3$. There in a unique set $\{X\}$ in $\bb{S}_{2,3,\{1,2\}}$. Therefore, the corresponding set in 
$\bb{S}_{3,2,\{2,3\}}$ must be $\{\bar{X}\} = \{R,G\}\smallsetminus \{X\}$. We pick $\tau=\{(1,B),(2,X),(3,x_3)\}$, and, indeed, $\{(2,X),(3,\bar{X})\}$ satisfies $X\neq \bar{X}$ while the output of $p_2$ is equal to its input. Therefore $\{(2,X),(3,\bar{X})\}\in\Cl(\Delta(\tau))$, which establishes~P1.  

\subparagraph{Input-output specification.} 

$\Delta'$ satisfies the following:

\begin{description}
\item[P3:] For every two simplices $\sigma = \{(i,x_i) : i \in I \} \in \m{I}$, and $\tau = \{(i,(x'_i,\bb{S}_i)) : i \in I\} \in \m{O}'$, where $I \subseteq [n]$, we set: 
$
    \tau \in \Delta'(\sigma) \iff \forall i \in I, x'_i = x_i.
$
\end{description}

\subparagraph{Example.} 

Still for the task of Fig.~\ref{fig:exampleIOP}, the simplex $\{(1,(B,\bb{S}_1)), (2,(X,\bb{S}_2)), (3,(x_3,\bb{S}_3))\}$ is a valid output for the input simplex $\{(1,B),(2,X),(3,x_3)\}$. This yields a simplicial map $\alpha:\m{I}\to\m{O}'$ which agrees with~$\Delta'$. Therefore, $(\m{I},\m{O}',\Delta')$ is solvable in zero rounds.  This is in agreement with our main result established in the next subsection.

\subsection{General Speedup Theorem}

We show that, under certain conditions on a task $(\m{I},\m{O},\Delta)$, the task $(\m{I},\m{O}',\Delta')$ where $\m{O}'$ is the complex defined by properties~P0-2, and $\Delta'$ is the input-output relation defined by~P3, is a speedup task of $(\m{I},\m{O},\Delta)$. The statements in this section are not entirely formal, as some additional properties are required for the results to hold. Nevertheless, these properties are essentially technical, and they do not impact the general message delivered by the statements below. For more details, see Sections~\ref{sec:t_to_t-1} and~\ref{sec:t-1_to_t}. 

\subsubsection{From $t$ rounds to $t-1$ rounds} 

\begin{lemma}\label{informal-weakspeedup}
Let $(\m{I},\m{O},\Delta)$ be a task, let $\m{M}$ be a model, let $t\geq 1$ be an integer, and let us assume that $(\m{I},\m{O},\Delta)$ satisfies the $(t-1)$-independence property w.r.t.~$\m{M}$. If $(\m{I},\m{O},\Delta)$ is solvable in $t$ rounds in~$\m{M}$, then the task $(\m{I},\m{O}',\Delta')$ is solvable in $t-1$ rounds in~$\m{M}$. 
\end{lemma}

\noindent\textbf{Sketch of Proof.} 
Let $\delta : \m{P}^{(t)} \to \m{O}$ solving $(\m{I},\m{O},\Delta)$  in $t$ rounds in~$\m{M}$. We define $\alpha:\m{P}^{(t-1)} \to \m{O}$. Note that, for any two  closed simplices $\varphi$ and $\psi$ in $\St(i,J_i)$, we have $\Sk_i (\Xi (\St(\sigma),\varphi)) = \Sk_i (\Xi (\St(\sigma),\psi))$ for every simplex~$\sigma\in\m{P}^{(t-1)}$. Therefore, we can abuse notation by denoting $\Sk_i (\Xi (\St(\sigma),\varphi))$ as $\Sk_i (\Xi (\St(\sigma),J_i))$. For any vertex $(i,v_i) \in \m{P}^{(t-1)}$, we let
$
\alpha(i,v_i)= (i,\bb{P}_i)  \;\mbox{with} \; \bb{P}_i = (x_i,\bb{S}_i),
$
where $x_i$ is the input value of process~$i$ (which is present in its view~$v_i$), and 
\[
\bb{S}_i = \{\bb{S}_{i,k,J_i} : ((i,J_i) \in \m{M}) \land (k \in [n]) \},
\]
where 
\[
\ora{\bb{S}_{i,k,J_i} =  \{ \delta(\Sk_i( \Xi (\St(\sigma),J_i))) : (\sigma \in \Sk_{\{i,k\}}(\St(i,v_i))) \land (\dim(\sigma) = |\{i,k\}| - 1) \}.}
\]
For every $\sigma \in \Sk_{\{i,k\}}(\St(i,v_i))$ with $\dim(\sigma) = |\{i,k\}| - 1$, the set 
$
S_{i,k,J_i}^\sigma  =  \delta(\Sk_i (\Xi (\St(\sigma),J_i)))
$
is the set of every possible output for process~$i$ using $\delta$ whenever the communication pattern~$J_i$ occurred at time~$t$, and the process $k$ has its  value fixed according to $\sigma\in\m{P}^{(t-1)}$. In particular, we have $S_{i,k,J_i}^\sigma \in 2^{\Sk_i(\m{O})}$, and thus $\bb{S}_{i,k,J_i} \in 2^{2^{\Sk_i(\m{O})}}$, as desired. 

We show that P0 holds. For every vertex $(i,J_i) \in \m{M}$, let us consider a set 
$
\{S_{i,k,J_i} \in \bb{S}_{i,k,J_i} : k \in J_i \}.
$
By definition, for every set $S_{i,k,J_i}$, there exists $\sigma_k \in \St(i,v_i)$ such that 
$
S_{i,k,J_i} = S_{i,k,J_i}^{\sigma_k} = \delta(\Sk_i(\Xi(\St(\sigma_k),E_i))). 
$
Using the $(t-1)$-independence property for $k = i \in J_i$, it holds that,
$
\bigcup_{k \in J_i \setminus \{i\} } \sigma_k \in \m{P}^{(t-1)}.
$
There is only one candidate for the simplex  corresponding to $k=i$, and this simplex is $\sigma_{i}=\{(i,v_i)\}$. Therefore,
$
\bigcup_{k \in J_i \setminus \{ i \} } \sigma_k = \bigcup_{k \in J_i} \sigma_k \in \m{P}^{(t-1)}.
$
For every $k \in J_i$, we have 
$
\Xi(\St(\cup_{k \in J_i} \sigma_k),J_i) \subseteq \Xi(\St(\sigma_k),J_i).
$
Therefore,
$
\delta(\Sk_i(\Xi(\St(\cup_{k \in J_i} \sigma_k),J_i))) \subseteq \bigcap_{k \in J_i} S_{i,k,J_i}. 
$
Now, $\delta(\Sk_i(\Xi(\St(\bigcup_{k \in J_i} \sigma_k),J_i)))$ cannot be empty, simply because $\bigcup_{k \in J_i} \sigma_k \in \m{P}^{(t-1)}$. Therefore, property P0 holds, that is, $\alpha$ produces vertices of~$\m{O}'$.

To prove that $\alpha$ solves $(\m{I},\m{O}',\Delta')$, it is sufficient to consider an arbitrary facet $\rho = \{(i,v_i) : i \in [n] \} \in \m{P}^{(t-1)}$, and its image $\alpha(\rho)=\{ (i,(x_i,\bb{S}_i) : i \in [n] \}$, and we show that $\alpha(\rho)$ is a facet of~$\m{O}'$ that agrees with~$\Delta'$. It is sufficient to show that both properties~P1 and~P2 hold as, by definition of $\alpha$, P3~holds by construction. 

First we prove that P1 holds. For every closed simplex $\varphi = \{(i,J_i) : i \in I \} \in \m{M}$, for every $(i,k) \in I \times I$ with $j \in J_k$ or $k \in J_i$, we consider the face $\sigma = \{(i,v_i),(k,v_k)\}$ of~$\rho$. Note that $\sigma \in \Sk_{\{i,k\}}(\St(i,v_i)) \cap \Sk_{\{i,k\}}(\St(k,v_k))$. Let us consider the sets 
\[
     S_{i,k,J_i}^\sigma =  \delta(\Sk_{i}(\Xi(\St(\sigma),J_i))) \in \bb{S}_{i,k, J_i} \;\mbox{and}\; 
     S_{k,i,J_k}^\sigma =  \delta(\Sk_{k}(\Xi(\St(\sigma),J_k))) \in \bb{S}_{k,i, J_k}. 
\]
Note that the set $S_{i,k,J_i}^\sigma$ (resp., $S_{k,i,J_k}^\sigma$) is independent of~$J_k$ (resp., $J_i$), by construction. 
It can be shown, again using $(t-1)$-independence, that P1 holds for these sets. 

Finally, P2 holds, also using $(t-1)$-independence. (See complete proof in Section~\ref{sec:t_to_t-1}).
\qed
\medbreak 

Note that Lemma~\ref{informal-weakspeedup} does not requires local checkability, and may therefore be applied even to tasks such as consensus in \wf. 

\subsubsection{From $t-1$ rounds to $t$ rounds} 

Our reciprocal of Lemma~\ref{informal-weakspeedup}, which guarantees that the task $(\m{I},\m{O}',\Delta')$  can be used for deriving a lower bound for $(\m{I},\m{O},\Delta)$, requires the task to satisfy a specific property, called \emph{edge-checkability}. The following definition is inspired from the notion of local checkability defined in~\cite{FraigniaudRT13} for the \wf\/ model.  Given a simplex $\sigma=\{(i,x_i):i\in I\}$, and $J\subseteq I$, we define 
$
\pi_J(\sigma)=\{(i,x_i):i\in J\}.
$
A task $(\m{I},\m{O},\Delta)$ is \emph{locally checkable} for the communication model $\m{M}$ if, for every $\sigma\in\m{I}$ with $\name(\sigma)=I\subseteq [n]$, for every set $\tau=\{(i,y_i):i\in I \land y_i\in\val(\m{O})\}$, and for every closed simplex $\varphi=\{(i,J_i):i\in I\}$ of $\m{M}$, the following holds: 
\[
\tau\in\Delta(\sigma) \iff \forall i\in I, \; \pi_{J_i}(\tau)\in \Delta (\pi_{J_i}(\sigma)).
\]
Moreover, the task $(\m{I},\m{O},\Delta)$ is \emph{edge-checkable} if  the following holds: 
\[
\tau\in\Delta(\sigma) \iff \forall i\in I, \; \forall k \in J_i, \; \pi_{\{i,k\}}(\tau)\in \Delta (\pi_{\{i,k\}}(\sigma)).
\]
Note that edge-checkability implies local checkability. For instance, renaming is edge-checkable in \wf, and proper coloring is edge-checkable in \local. In particular, the task depicted on  Fig.~\ref{fig:exampleIOP}  is edge-checkable. On the other hand, consensus is not even locally checkable in \wf, and the standard version of maximal independent set (MIS), where a node in the set is labeled~1, while a node not in the set is labeled~0, is locally checkable in \local, but not edge-checkable in \local. The ``edge version'' of MIS defined in~\cite{Balliu0HORS19,Brandt19} is however edge-checkable. As in~\cite{Brandt19}, we assume an underlying mechanism enabling the any two processes $i$ and $k$ such that $i\in J_k$ or $k\in J_i$  to \emph{break symmetry}.

\begin{lemma}\label{informal-recipro-speedup}
Let $(\m{I},\m{O},\Delta)$ be a task, let $\m{M}$ be a model, let $t\geq 1$ be an integer, and let us assume a symmetry-breaking mechanism, and that $(\m{I},\m{O},\Delta)$ is edge-checkable in~$\m{M}$. If $(\m{I},\m{O}',\Delta')$ is solvable in $t-1$ rounds in~$\m{M}$ then $(\m{I},\m{O},\Delta)$~is solvable in $t$ rounds in~$\m{M}$. 
\end{lemma}

\noindent\textbf{Sketch of Proof.} 
    It is  sufficient to show the existence of a simplicial map~$\beta:\Xi(\m{O}')\to \m{O}$ such that, for every closed simplex $\sigma\in \m{I}$, $\beta(\Xi(\Delta'(\sigma)))\subseteq \Delta(\sigma)$. 
    
    Let $\tau = \{{(i,(x_i,\bb{S}_i)) : i \in I} \}\in\Delta'(\sigma)$ for some closed simplex $\sigma = \{(i,x_i) : i \in I \} \in \m{I}$. Note that $\tau$ is a face of a facet of $\m{O}'$, which, by definition, satisfy P1 and P2. Since $(\m{I},\m{O},\Delta)$ is edge-checkable, it is sufficient to prove that, for every closed simplex $\varphi=\{(i,J_i):i\in I\} \in \m{M}$, for every $i\in I$, process~$i$ can output a solution $(i,y_i) \in \Sk_i(\m{O})$ such that,  for every $k \in J_i$, $\{(i,y_i),(k,y_k)\} \in \Delta(\pi_{\{i,k\}}(\sigma))$. 
    
    After one round of communication according to~$\varphi$, every  process~$i\in I$ receives messages from every process $k \in J_i$. Therefore every process~$i$ has access to the set $\{ (x_k,\bb{S}_k) : k \in J_i \}$. If $k \in J_i$ does not receive from~$i$ then property~P2 guarantees that $|\bb{S}_{k,i,J_k}| = 1$. If $k \in J_i$ receives from~$i$ (i.e $i \in J_k$ then the symmetry-breaking mechanism, and property~P1 allow the process $i$ and $k$ to choose sets $(S_{i,k,J_i},S_{k,i,J_k}) \in \bb{S}_{i,k,J_i} \times \bb{S}_{k,i,J_k}$ such that, for every choice of $(y_i,y_k)$ where $(i,y_i)\in S_{i,k,J_i}$ and $(k,y_k)\in S_{k,i,J_k}$, there exists $\tau \in \m{I}$ such that
    \[
    \{(i,x_i),(k,x_k)\} \subseteq \tau, \; \mbox{and} \;
    \{(i,y_i),(k,y_k)\} \in \Delta(\tau). 
    \]
    Since $(\m{I},\m{O},\Delta)$  is edge-checkable, $\{(i,y_i),(k,y_k)\} \in \Delta(\tau)$ implies that 
    \[
    \{(i,y_i),(k,y_k)\} \in \Delta(\pi_{\{i,k\}}(\tau)) = \Delta(\pi_{\{i,k\}}(\sigma)).
    \]
    By repeating this operation for every process $k \in J_i$, process~$i$ can output any value $y_i$ such that,
    \[
        (i,y_i) \in \bigcap_{k \in J_i} S_{i,k,J_i}. 
    \]
    Such a value $y_i$ does exist thanks to property~P0. The correctness of this algorithm is straightforward since the task is edge-checkable, and the sets $S_{i,k,J_i}$ are precisely chosen to satisfy~P1.
\qed

\medbreak

Using Lemmas~\ref{informal-weakspeedup} and~\ref{informal-recipro-speedup}, we immediately derive our main result. 

\begin{theorem}\label{informal-strspeedup}
Let $(\m{I},\m{O},\Delta)$ be a task, let $\m{M}$ be a model, let $t\geq 1$ be an integer, and let us assume that $(\m{I},\m{O},\Delta)$ satisfies the $(t-1)$-independence property w.r.t.~$\m{M}$, and is edge-checkable in~$\m{M}$. $(\m{I},\m{O},\Delta)$ is solvable in $t$ rounds in~$\m{M}$ if and only if $(\m{I},\m{O}',\Delta')$ is solvable in $t-1$ rounds in~$\m{M}$.
\end{theorem}

\subsection{Applications}

We illustrate the generality of Theorem~\ref{informal-strspeedup} by examples from two radically different settings, namely  shared-memory wait-free computing, and synchronous failure-free network computing. 

\subsubsection{Shared-Memory Wait-Free Computing}
\label{sec:appwf}

We consider the standard \emph{perfect renaming} task in \wf, with two processes. The two processes starts with distinct identifiers in $\{0,1,2\}$ as input, and they are asked to output a distinct identifiers in $\{0,1\}$. We provide a new impossibility proof for perfect renaming, using Theorem~\ref{informal-strspeedup}.

\begin{corollary}\label{cor:renaming}
Perfect renaming in 2-process system is impossible in \wf. 
\end{corollary}

\begin{proof}
We start with two observations. First, for any $t\geq 0$,  renaming satisfies $t$-independence w.r.t. \wf\/ with two processes. Second \emph{perfect renaming} is \emph{edge-checkable} in \wf. Indeed, in \wf, if every process receives identifiers that are different from its own identifier, then all identifiers are necessarily distinct. Therefore, Theorem~\ref{informal-strspeedup} applies. Let us identify the task  $(\m{I},\m{O}',\Delta')$ defined by Properties P0-3 applied to perfect renaming (the input and output complexes of perfect renaming, $\m{I}$ and $\m{O}$, are displayed on Fig.~\ref{fig:examplerenaming}, and the input-output specification is trivial, i.e., $\tau\in\Delta(\sigma)$ whenever $\name(\tau)=\name(\sigma)$). 

In the general construction, we have $\bb{S}_j \in 2^{2^{\Sk_{\{j\}}(\m{O})}}$ for process~$j$. However, for the sake of simplifying the notations,  we manipulate sets in $2^{2^{\val(\m{O})}}=2^{2^{\{0,1\}}}$.
Any set $S \in 2^{\val(\m{O})}$ must be one of the following three sets: $\mb{0} = \{0\}, \mb{1} = \{1\},$ and $\mb{X} = \{0,1\}$. Indeed, the empty set $S=\varnothing$ does not satisfy~P0.  
We denote by $p_i$ and $p_k$ the two processes in the systems (instead of $p_1$ and $p_2$, for avoiding confusion between process indexes and input and output values). 
Let us consider a facet $\{(\ell,(x_\ell,\mathbb{S}_\ell)) : \ell \in \{i,k\} \} \in \m{O}'$. By definition,  for any $\ell \in \{i,k\}$, we have 
\[
\mathbb{S}_\ell = \{\mathbb{S}_{\ell,\ell',J_k} : (\ell,J_\ell) \in \mathcal{M} \land \ell' \in \{i,k\}\}.
\]
Let us focus on $\mathbb{S}_{i,k,J}$ and $\mathbb{S}_{k,i,J}$ for $J = \{i,k\}$. Recall that any $S \in \mathbb{S}_{i,k,J}$ (resp., $\mathbb{S}_{k,i,J}$) is non-empty thanks to the universal quantifier in~P0. Moreover, $\mathbb{S}_{i,k,J}$ (resp., $\mathbb{S}_{k,i,J}$)  is itself non-empty thanks to the existential quantifier in~P1. (These two facts actually hold for any set in $\bb{S}_i$ or $\bb{S}_k$.)
By Property~P1, there exists $(S_{i,J},S_{k,J}) \in \mathbb{S}_{i,k,J} \times \mathbb{S}_{k,i,J}$ such that, for every $(y_i,y_k) \in S_{i,J} \times S_{k,J}$,
\[
    \{(i,y_i),(k,y_k)\} \in \Delta(\{(i,x_i),(k,x_k)\}) = \m{O}.
\]
We necessarily have $S_{i,J}\neq \mb{X}$ because, for any $y \in \{0,1\}$, either $\{(i,0),(k,y)\} \notin \m{O}$, or $\{(i,1),(k,y)\} \notin \m{O}$, and therefore it is not possible that both sets are simplices of $\m{O}$. By the same arguments, we also have $S_{k,J}\neq \mb{X}$. 
It follows that $\mathbb{S}_{i,k,J}$ and $\mathbb{S}_{k,i,J}$ can only take three possible values : $\bb{O} = \{\mb{0}\}$, $\bb{I} = \{\mb{1}\}$ and $\bb{X} = \{\mb{0},\mb{1}\}$. 

By the same arguments applied on the facet  $\{(i,\bb{S}_{i,k,J}),(k,\bb{S}_{k,i,\{k\}})\}$, we get that $\bb{S}_{k,i,\{k\}}$  can also only take its values in $\{\bb{O},\bb{I},\bb{X}\}$. 
By Property~P2, it must be the case that $\bb{S}_{k,i,\{k\}} = \bb{S}_{k,k,\{k\}}$. On the other hand,  $|\bb{S}_{k,k,\{k\}}| = 1$ implies that $\bb{S}_{k,i,\{k\}} \in \{\bb{O},\bb{I}\}$. 

Symmetrically the same  holds for process~$i$, that is,  $\bb{S}_{i,k,\{i\}} \in \{\bb{O},\bb{I}\}$.

We now show that, necessarily, $\bb{S}_{i,k,\{i\}} \neq \bb{S}_{k,i,\{k\}}$. Let us consider the two sets
 $S_{i,\{i\}} \in \bb{S}_{i,k,\{i\}}$ and $S_{k,\{k\}} \in \bb{S}_{k,i,\{k\}}$ such that $(S_{i,\{i\}},S_{k,J})$ and $(S_{k,\{i\}},S_{i,J})$ both satisfy~P1. 
Since $(S_{i,J},S_{k,J})$ satisfy~P1, and since $(S_{i,J},S_{k,J})\in \{\bb{O},\bb{I}\}^2$,  we have $S_{i,J} \neq S_{k,J}$. It follows that $S_{k,J} \neq S_{i,\{i\}}$, and therefore $S_{i,J} = S_{i,\{i\}}$. Similarly, we have  $S_{i,J} \neq S_{k,\{i\}}$, and therefore $S_{k,J} = S_{k,\{k\}}$. As a consequence, $S_{i,\{i\}} \neq S_{k,\{k\}}$, and thus $\bb{S}_{i,k,\{i\}} \neq \bb{S}_{k,i,\{k\}}$, as claimed.

Overall, we have shown that $(\bb{S}_{i,k,\{i\}},\bb{S}_{k,i,\{k\}}) \in \{(\bb{O},\bb{I}),(\bb{I},\bb{O})\}$. Therefore, by replacing $\bb{O}$ by~0, and $\bb{I}$ by~1, there is a one-to-one correspondence between the partial outputs $(\mathbb{S}_{i,k,\{i\}},\mathbb{S}_{k,i,\{k\}})$ for $(\m{I},\m{O}',\Delta')$, and valid outputs for perfect renaming. It follows that if there is an algorithm for solving the task $(\m{I},\m{O}',\Delta')$ defined by Properties P0-3 applied to perfect renaming, then, in particular, this algorithm also solves perfect renaming. 
Therefore, thanks to Theorem~\ref{informal-strspeedup}, we get that, for every $t\geq 1$,  if perfect renaming is solvable in $t$~rounds in \wf, then perfect renaming is solvable in $t-1$ rounds in \wf. Since perfect renaming is not solvable in zero rounds,  we conclude that perfect renaming is not solvable in \wf.
\end{proof}

\subparagraph{Remark.} 

The proof of Corollary~\ref{cor:renaming} illustrates a quite interesting case, where solving the speedup task $(\m{I},\m{O}',\Delta')$ obtained using P0-3 includes solving the original task $(\m{I},\m{O},\Delta)$. In this case, Theorem~\ref{informal-strspeedup} is not necessary, and Lemma~\ref{informal-weakspeedup} suffices for establishing the impossibility of solving the task $(\m{I},\m{O},\Delta)$. That is, the edge-checkability condition is not required, and solely local-independence is required. An interesting application is 2-process consensus, which is not edge-checkable in \wf. Nevertheless, the speedup task of consensus obtained using P0-3 happens to include consensus itself,  by the same type of arguments as in the proof of Corollary~\ref{cor:renaming}. Impossibility of consensus therefore directly follows from Lemma~\ref{informal-weakspeedup}. 

\begin{corollary}\label{cor:consensus}
Consensus in 2-process system is impossible in \wf. 
\end{corollary}

\subsubsection{Synchronous failure-free network computing}

As mentioned before, several models can satisfy the conditions in the statement of Theorem~\ref{informal-strspeedup}, beyond \local. This is, for instance, the case for some dynamic graph models, $\dyn(\m{F})$, where, at each round, one of the graphs in $\m{F}$ is chosen to be the underlying communication graph. This is also the case for some hypergraph models, $\hlocal(G)$, which is the natural extension of \local\/ to hypergraphs~$G$ (\textsf{H} stands for hypergraph). It can be shown that, in \hlocal, edge-checkability  is essentially equivalent, up to $\pm 1$ rounds, to  local checkability (see Lemma~\ref{lem:locdecvsedgdec}). 
The following result is a direct consequence of our generalised version of Theorem~\ref{informal-strspeedup} (see Theorem~\ref{strspeedup}). Interestingly, as communications produced by hypergraphs could be viewed as communications on a graph where hyperedges are transformed into cliques, the result below also shows that the ``large girth property'' is, to some extend,  not necessary for applying Theorem~\ref{informal-strspeedup}. For instance, Theorem~\ref{informal-strspeedup} could be applied to communication graphs such as the one represented in Figure~\ref{fig:hypertree}. The proof of the following can be found in Section~\ref{sec:application2hypergraphs}.

\begin{figure}[tb]
    \centering
     \includegraphics[scale = 0.1]{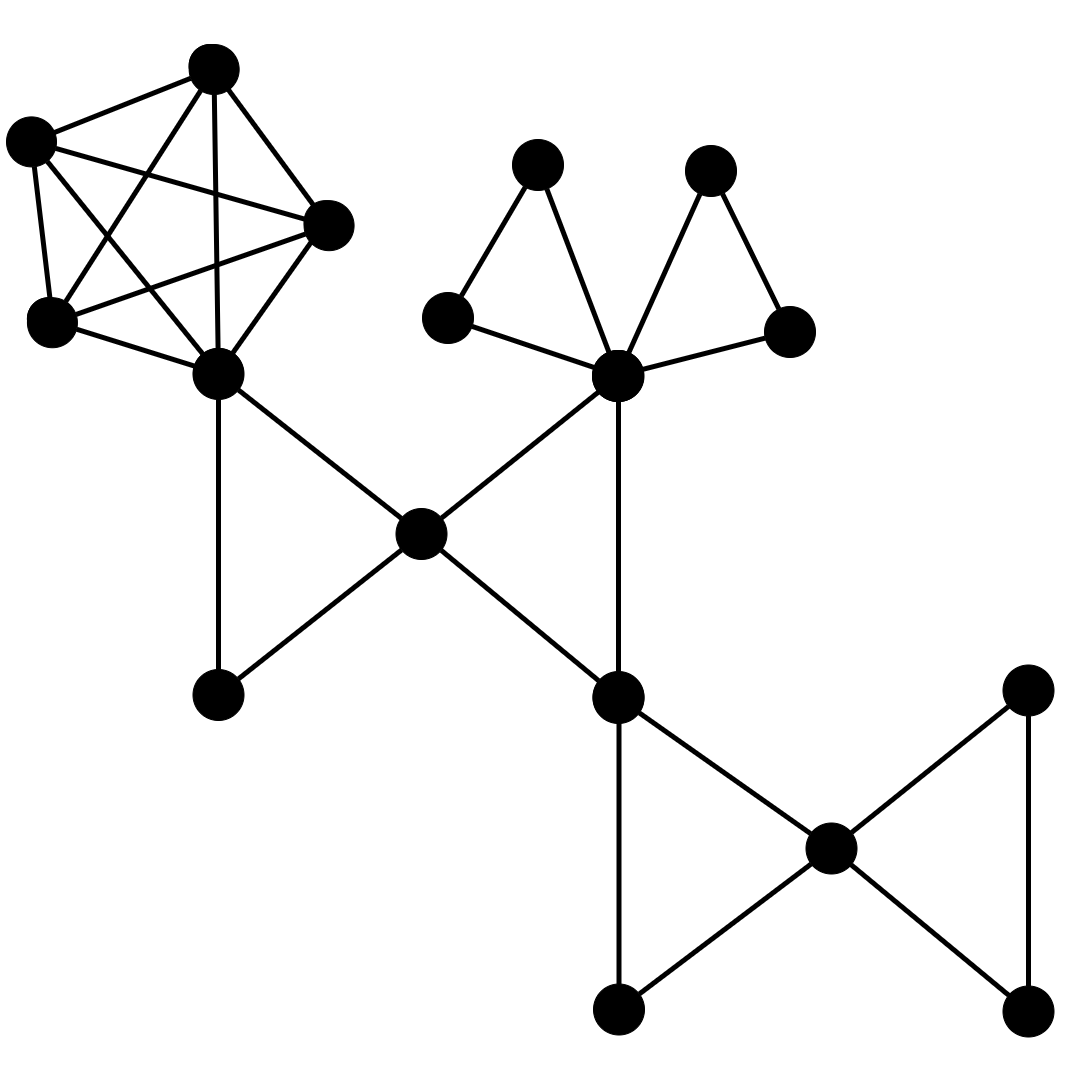}  
    \caption{A graph induced by a linear hypertree}
    \label{fig:hypertree}
\end{figure}

\begin{corollary}
Let  $H$ be a hypergraph on~$n$ nodes, let $\Pi=(\m{I},\m{O},\Delta)$ be a task for $n$ processes, and let $t\geq 1$. Let us assume that $\Pi$ satisfies the $r$-independence property w.r.t.~$\hlocal(H)$ for every $r\in\{0,\dots,t-1\}$, and that $\Pi$ is edge-checkable in~$\hlocal(H)$. Let us  assume the existence of a symmetry breaking mechanism among the processes in each hyperedge of~$H$.  
$\Pi$ is solvable in $t$ rounds in~$\hlocal(H)$ if and only if $\Pi^{(t)}=(\m{I},\m{O}^{(t)},\Delta^{(t)})$ is solvable in zero rounds in~$\hlocal(H)$, where $\Pi^{(t)}$ is the task defined by iterating $t$~times the construction defined by properties~P0-3.
\end{corollary}

\subsection{Technical Summary}

Brandt's speedup theorem holds thanks to an operator~$\Phi$ which, given any task $\Pi=(\m{I},\m{O},\Delta)$, constructs a task $\Phi(\Pi)=(\m{I},\m{O}',\Delta')$ such that, for every $t\geq 1$, if $\Pi$ is $t$-independent w.r.t. the \local\/ model, and if $\Pi$ is edge-checkable in the \local\/ model then 
\[
\Pi \;\mbox{is solvable in $t$ rounds in \local} \iff \Phi(\Pi) \;\mbox{is solvable in $t-1$ rounds in \local.} 
\]
We have extended Brandt's operator $\Phi$ to an operator $\Phi^\star$ defined by Properties P0-3, which applies to any round-based model~$\m{M}$ supporting full-information protocols. We have extended the notion of local independence and edge-checkability to these models, which allows us to extend the equivalence above to all such communication model~$\m{M}$, using $\Phi^\star$ instead of $\Phi$. Interestingly, the operator $\Phi^\star$ is applicable even to asynchronous models such as \wf, and allows us to provide new impossibility proofs for consensus and renaming in 2-process systems. 
Defining an operator $\Phi_{\mathsf{WF}}$ transforming any task $\Pi=(\m{I},\m{O},\Delta)$ into a task $\Phi_{\mathsf{WF}}(\Pi)=(\m{I},\m{O}',\Delta')$ such that, under some conditions, for every $t\geq 1$ and $n\geq 2$, $\Pi$ is solvable in $t$ rounds in \wf\/ if and only if  $\Phi_{\mathsf{WF}}(\Pi)$ is solvable in $t-1$ rounds in \wf. Our operator $\Phi^\star$ applies to \wf\/ for $n=2$. Extending $\Phi^\star$ to some $\Phi_{\mathsf{WF}}$ for arbitrary~$n$ requires to overcome the fact that \wf\/ does not satisfy local independence for $n>2$. 

The remaining of the paper formalizes and generalizes the concepts and ideas provided in this section. 

\section{A General Model of Communication}
\label{sec:general-framework}

This section describes our framework. We assume that the reader is familiar with the basic concepts of combinatorial topology applied to distributed computing, and in particular with the notion of \emph{task} $(\m{I},\m{O},\Delta)$. The  reader unfamiliar with these concepts may refer to Appendix~\ref{sec:distributed-tasks}. 
Our framework encapsulates several standard communication media, from asynchronous crash-prone shared-memory computing to synchronous failure-free network computing, such as: 
\begin{itemize}
\item $\wf(n)$ refers to the asynchronous shared-memory model involving $n$ crash-prone processes interacting via iterated immediate snapshots~\cite{AttiyaW04}. Note that this model is computationally equivalent to the model with atomic read/write operations~\cite{Herlihy2013}. 
$\tr(n)$ refers to the same model as $\wf(n)$, except that at most $f$ processes can crash, where $1\leq f \leq n-1$~\cite{AttiyaW04}. 

\item $\local(G)$ refers to the synchronous failure-free message-passing model in networks~\cite{Peleg00}, i.e., processes are nodes of the (undirected) graph~$G$, computation and communication proceed in lockstep, and a message is exchanged along each edge of~$G$ between neighboring processes at every time step. This model can trivially be extended to $\hlocal(H)$, where $H$ is a hypergraph. 

\item $\dyn(\m{F})$ refers to the dynamic network model, which is the same as \local\/ except that the communication graph may evolve with time~\cite{CasteigtsFQS11,Ferreira04}. Specifically, given a family $\m{F}$ of $n$-node graphs on the same set of vertices, the communications occurring at any step are performed along the edges of one of the graphs in~$\m{F}$. Again, this model can trivially be extended to $\hdyn(\m{F})$ where $\m{F}$ is a family of hypergraphs. 
\end{itemize}

All these models share properties that are essential to our framework. First, they all include a notion of \emph{round}, either explicitly like in \local\/ and \dyn, or implicitly like in \wf\/ and \tr. Second, they all support \emph{full information} protocols, that is, whenever a process communicates, whether it be by sending/receiving messages to/from neighbors, or writing/reading in a shared memory, it communicates its entire history since the beginning of its execution. This assumption enables the design of \emph{robust} lower bounds, which still hold if the communication medium restricts the communication power somehow. Last but not least, none of these models place limitations on the individual computing power of each process, which enables the design of \emph{unconditional} lower bounds or impossibility results. 
Note that none of these models refer to IDs. The fact that nodes may or may not be provided with IDs is, in our framework, not a property of the communication model, but a property of the tasks to be solved in this model. This is formally specified below. 

\subparagraph{Names and Identifiers.}

Identifiers can be viewed as forming a special type of input values, for which it is assumed that every two different processes are assigned different IDs. We stress the fact that \emph{the ID of a process must not be confused with its name}. The former  is a value given to the process as input, while the latter is external, whose sole purpose it to refer to the process. In this paper, we always assume that \emph{processes are not aware of their names}, and \emph{processes may or may not be granted with IDs} as part of their inputs. 
 For $i\in [n]$, process~$i$ is denoted by $p_i$, and its identifier (if any) by $\id(p_i)$. It is systematically the case that $\name(p_i)=i$, while $\id(p_i)$ is some arbitrary value taken in some finite set of integers. When IDs are assigned to the processes, the input value of $p_i$ is therefore a pair $(x,y)$ where $x$ is the ID of~$p_i$, and $y$ is the input \emph{label} of $p_i$. Hence, a vertex of the input complex is of the form 
 $
 (i,s)=\big(\!\name(p_i),\val(p_i)\big)=\Big(\!\name(p_i),\big(\id(p_i),\lab(p_i)\big)\Big).
 $
A task $(\m{I},\m{O},\Delta)$ does not necessarily involve IDs in its specification, e.g., consensus. However, given a tasks $(\m{I},\m{O},\Delta)$ with no IDs, solving $(\m{I},\m{O},\Delta)$ with IDs in $[1,N]$ is merely the task  $(\m{I}',\m{O},\Delta')$ where $\{(i,(x_i,v_i)):i\in I\}$ is in~$\m{I}'$  if  (1)~$\{(i,v_i):i\in I\}$ is in~$\m{I}$, and, (2)~for every $i\in I$, $x_i\in [1,N]$ with  $x_i\neq x_j$ for every $j\in I\smallsetminus\{i\}$.  Moreover, for every $\sigma'=\{(i,(x_i,v_i)):i\in I\}\in \m{I}'$ and $\tau\in\m{O}$, we set $\tau\in\Delta'(\sigma')$ whenever $\tau\in\Delta(\sigma)$ where $\sigma=\{(i,v_i):i\in I\}\in \m{I}$. 

\subparagraph{Communication Models.}

We now describe a way to encapsulate various communication models in a single general framework. Every  communication medium is modeled as a complex $\m{M}$. To describe this complex, let $\m{S}_n$ be the $(n-1)$-dimensional chromatic pseudo-sphere induced by~$2^{2^{[n]}}$ (cf. Appendix~\ref{app:element-topo}). A vertex of $\m{S}_n$ is a pair $(i,E)$ with $i\in [n]$, and $E\subseteq 2^{[n]}$. The semantics of such a vertex is that $p_i$ receives information from all  $p_j, j\in e$, for every $e\in E$.  Given a vertex $(i,E)$ of $\m{S}_n$, each element $e\in E$ is called a \emph{channel}.
In fact, we will consider only vertices $(i,E)$ where $\{i\}\in E$, that is, we systematically assume that a process has a (private) channel to itself, a.k.a.~a ``self-loop''. Also, w.l.o.g., for simplifying and unifying the presentation, we will consider only vertices $(i,E)$ where, for every $e\in E$, $i\in e$. 

\begin{definition}\label{Def:comm} 
A communication model is a pure $(n-1)$-dimensional sub-complex $\m{M}$ of~$\m{S}_n$ such that, for every vertex $(i,E)\in \m{M}$, $\{i\}\in E$, and, for every $e\in E$, $i\in e$.
\end{definition}

In many classical models, every channel has cardinality~2 (putting aside the self-loops). Given a vertex $(i,E)$, and a channel $e=\{i,j\}\in E$, the semantic of this channel is that $p_i$ receives information from~$p_j$ via the channel~$e$. Whenever all channels have cardinality~$\leq 2$, the vertices of~$\m{M}$ can merely be represented as pairs $(i,E)$ where $E\subseteq [n]$, and $i\in E$. In that case, an element $j\in E$ represents a channel $\{i,j\}$. It is however worth considering the case where channels have larger cardinalities, as it actually helps in the statement of our results, and it makes these results directly applicable to models involving multiparty communications, e.g., involving hypergraphs (a channel is then merely a hyperedge). 

The semantics of a simplex $\{(i,E_i):i\in I\}$, $I\subseteq [n]$, of $\m{M}$ is that the model allows an instance of communication in which,  during a \emph{same} round, every process $i\in I$ receives a message from all processes in the channels belonging to~$E_i$.  Here are a few examples, where the first three deal with channels involving two processes (therefore the sets~$E_i$ are merely subsets of~$[n]$), while the fourth one involves potentially larger channels. 

\begin{itemize}
\item $\wf(n)$: for every non-empty $I\subseteq [n]$, a set $\{(i,E_i):i\in I\}$ is a simplex of $\m{M}$ if, for every $i,j\in I$, we have 
$E_i\subseteq [n]$, ($E_i\subseteq E_j \lor E_j\subseteq E_i$), and we have
$E_j\subseteq E_i$ whenever $j\in E_i$. $\m{M}$ is actually isomorphic to the chromatic subdivision of the complete complex with vertex set $V=\{(i,\bot):i\in [n]\}$ (see~\cite{HerlihyS99}). 

\item $\tr(n)$: for every non-empty $I\subseteq [n]$, a set $\{(i,E_i):i\in I\}$ is a simplex of $\m{M}$ if, for every $i,j\in I$, we have 
$E_i\subseteq [n]$, ($E_i\subseteq E_j \lor E_j\subseteq E_i$),
$E_j\subseteq E_i$ whenever $j\in E_i$, and $|E_i|\geq n-f$. Indeed, since at most $f$ processes can crash, every process can wait for at least $n-f-1$ other processes before proceeding to the next round. 

\item $\local(G)$: given an $n$-node graph $G$ with nodes labeled from~1 to~$n$, for every non-empty $I\subseteq [n]$, a set $\{(i,E_i):i\in I\}$ is a simplex of $\m{M}$ if, for every $i\in I$, $E_i=N_G[i]$ where $N_G[i]$ denotes the closed neighborhood of node~$i$ in~$G$. In $\dyn(\m{F})$, given a family $\m{F}$ of $n$-node graphs with nodes labeled from~1 to~$n$, for every non-empty $I\subseteq [n]$, a set $\{(i,E_i):i\in I\}$ is a simplex of $\m{M}$ if there exists $G\in \m{F}$ such that, for every $i\in I$, $E_i=N_G[i]$. 

\item $\hlocal(H)$: given an $n$-node hypergraph $H$ with nodes labeled from~1 to~$n$, for every non-empty $I\subseteq [n]$, a set $\{(i,E_i):i\in I\}$ is a simplex of $\m{M}$ if, for every $i\in I$, $E_i=E_H(i)$ where $E_H(i)$ denotes the set of hyperedges of~$H$ containing node~$i$. For the sake of technical uniformity, we assume, w.l.o.g., that $\{i\}\in E_H(i)$. 
\end{itemize}

Some models $\m{M}$ are \emph{non-deterministic}, in the sense that, given a set $I\subseteq [n]$ of processes, there are more than one simplices $\sigma\in \m{M}$ with $\name(\sigma)=I$. This is the case for  \wf, \tr, and \dyn. This reflects the fact that the pattern of communication may differ at each round, whether it be because the processes are asynchronous, because some processes have crashed, or because the communication network evolves with time. Instead, \local\/ and $\hlocal\/$ are deterministic in the sense that the communication pattern performed at each round is identical through time.
Beyond the standard models captured by our formalism, the formalism is flexible enough to capture a vast class of other models, including the following asynchronous wait-free variant of \local. 
\begin{itemize}
\item $\wflocal(G)$ denotes the model~$\m{M}$ in which, given an $n$-node graph $G$ with nodes labeled from~1 to~$n$,  for every non-empty $I\subseteq [n]$, a set $\{(i,E_i):i\in I\}$ is a simplex of $\m{M}$ if, for every $i\in I$, we have $E_i\subseteq N_G[i]$, and, for every  $U \subseteq I$, if the subgraph of $G$ induced by $U$ is a clique, then for every $i,j\in U$, 
$(E_j\cap U \subseteq E_i \cap U \lor E_i\cap U \subseteq E_j \cap U$), and 
$E_j\cap U \subseteq E_i$ whenever $j\in E_i$. (In particular, for every edge $\{i,j\}\in E(G)$, $j\in E_i$ or $i\in E_j$, or both).
\end{itemize}

\subparagraph{Open and Closed Simplices.}

There are two types of simplices in~$\m{M}$. A \emph{closed} simplex is a simplex  $\{(i,E_i):i\in I\}$ such that the set of processes in the union of all channels is~$I$, i.e., $\bigcup_{i\in I}\bigcup_{e\in E_i} e=I$. A simplex that is not closed is called \emph{open}. Closed simplices play an important role as they can be used to model various scenarios in which the processes in~$I$ are disconnected from the other processes, whether it be because the latter crashed, or because the processes in~$I$ are forming a connected component of a disconnected network. In $\wf(n)$, for every $I\subset [n]$, there are simplices $\sigma\in\m{M}$ with $\name(\sigma)=I$ that are closed, and there are simplices $\sigma\in\m{M}$ with $\name(\sigma)=I$ that are open. Instead, in \local\/ and \dyn, as long as the networks are connected,  only the facets of $\m{M}$ are closed, and all the lower dimensional simplices are open. 
Similarly, in $\tr(n)$, any simplex~$\sigma$ with  $|\name(\sigma)| < n-f$ is necessarily open. 

\subparagraph{Local Encoding of the processes and channels. }

In the presence of IDs assigned to the processes, a process~$i$ can trivially identify the other processes from which it receives information thanks to their IDs. In absence of IDs, and/or in the presence of large channels, we assume a mechanism baring similarities with the \emph{port-numbers} in anonymous variants of~\local~\cite{HirvonenS2020}. Specifically, for every process~$i$, each channel~$e$ incident to~$i$ is uniquely identified by~$p_i$, and every process~$j\in e$ is also unambiguously identified by~$p_i$, that is, if $j\in e\cap e'$ for two channels $e$ and $e'$ incident to~$i$, process~$i$ correctly identifies the information received from~$j$ as information received from a  same process. Moreover, all processes in a same channel $e$ identify $e$ with the same name.

\subparagraph{Communication Map and Protocol Complexes.}

We now describe the evolution of the system along with time when a full-information protocol is executed for solving a task. 
To every communication model $\m{M}$ corresponds a \emph{communication} map~$\Xi$ that applies to any chromatic complex~$\m{K}$, defined as follows. Let $\sigma=\{(i,v_i):i\in I\} \in \m{K}$, where $I\subseteq [n]$, and let us assume that there exists a closed simplex $\varphi=\{(i,E_i):i\in I\}$ in~$\m{M}$. We set 
$$
\Xi(\sigma,\varphi)=\Big \{\big (i,\big\{\{v_j:j\in e\}:e\in E_i\big \}\big ):i\in I\Big \}. 
$$
Note that, for every $i\in I$, and every $e\in E_i$, $\{v_j:j\in e\}$ is a multiset, and so is $\big\{\{v_j:j\in e\}:e\in E_i\big \}$. The elements of both multisets are indexed locally by process~$i$, thanks to the local identification of the channels and of the other processes. 
For every $\sigma\in\m{K}$, we define
$
\Xi(\sigma)=\{\Xi(\sigma,\varphi) : (\varphi \in \m{M})\wedge(\name(\varphi)=\name(\sigma))\wedge(\mbox{$\varphi$ is closed})\}.
$
The communication map $\Xi$ merely reflects the fact that if the processes in~$\sigma$ interact among themselves according to $\varphi\in\m{M}$, then they will end up in the state $\Xi(\sigma,\varphi)$ in $\Xi(\sigma)$. Note that if all simplices $\varphi\in \m{M}$ with $\name(\varphi)=\name(\sigma)$ are open, then 
$
\Xi(\sigma)=\varnothing.
$
An empty $\Xi(\sigma)$ reflects the fact that there are no communication patterns for which the processes in $\name(\sigma)$ communicate solely among themselves. We say that a simplex $\sigma\in\m{K}$ is closed in model~$\m{M}$ if there exists a closed simplex $\varphi\in \m{M}$ with $\name(\varphi)=\name(\sigma)$. The simplex $\sigma$ is open otherwise. 

\begin{definition}
Given a communication model~$\m{M}$, and given a pure $(n-1)$-dimensional chromatic complex~$\m{K}$, we define $\Xi(\m{K})$ as the closure of the images by~$\Xi$ of all the closed simplices of~$\m{K}$. In other words, 
$
\Xi(\m{K}) = \bigcup_{\sigma\in \m{K}, \; \sigma \; \mbox{\rm\footnotesize closed}}\Cl(\Xi(\sigma)).
$ 
\end{definition}

Note that, by definition of the closure operator, $\Xi(\m{K})$ is a complex. Moreover, this complex is pure, with dimension~${n-1}$. It is precisely the complex representing all possible states of the system after one round of communication starting from the state complex~$\m{K}$. The communication map~$\Xi$ induced by a model~$\m{M}$ enables to specify the evolution of the system along with the course of an execution starting from any initial state. 

\begin{definition}
Given the input complex $\m{I}$ of some task $(\m{I},\m{O},\Delta)$, the \emph{protocol complex} at time $t\geq 0$, denoted by~$\m{P}^{(t)}$, is defined as $\m{P}^{(0)}=\m{I}$, and, for every $t>0$, $\m{P}^{(t)}=\Xi(\m{P}^{(t-1)})$. 
 \end{definition}

By definition, a vertex of the protocol complex $\m{P}^{(t)}$ is a pair $(i,w)$ where $w$ is a possible \emph{view} of the system as perceived by process~$i$ at time~$t$. This view depends on the  history experienced by process~$i$ during $t$ rounds of communication, and is very much depending on the communication model. For instance, in $\local(G)$, the view~$w$ is a ball of radius~$t$ in the graph~$G$, whose nodes are labeled by their input values (and their IDs if the task at hand assumes identifiers).  In \wf, $w$ results from the sequence of $t$ immediate-snapshot instructions performed at the successive levels~1 to~$t$ of the shared memory, whose values depend on the interleavings of these instructions performed asynchronously by all the processes. By definition, a set $\{(i,w_i):i\in I\}$ of vertices of $\m{P}^{(t)}$, with $I\subseteq [n]$, forms a simplex of $\m{P}^{(t)}$ if the views $w_i$, $i\in I$, are mutually compatible, i.e., there is a sequence of $t$ communication rounds leading every process $i\in I$ to acquire the view~$w_i$. 

\subparagraph{Computation as Simplicial Maps.}

In the context of full-information protocols, the design of an algorithm boils down to computing an output at every process after a given number~$t\geq 0$ of communication rounds. That is, a $t$-round algorithm consists of (1)~communicating for $t$~rounds, and (2)~computing an output at each process. A $t$-round algorithm for a task $(\m{I},\m{O},\Delta)$ is therefore merely a chromatic (i.e., name-preserving) function 
$
\delta:V(\m{P}^{(t)})\to V(\m{O})
$
mapping every vertex $(i,s)\in V(\m{P}^{(t)})$ to some vertex $(i,v)\in V(\m{O})$. The semantic of this mapping is that process~$i$ in state~$s$  at round~$t$ outputs the value~$v$. For the algorithm to be correct, it must be the case that $\delta$ is \emph{simplicial}, i.e., it maps simplices to simplices. Indeed any possible global state of the system at time~$t$ must be mapped to some legal global output state.  Moreover, $\delta$ must  agree with the specification $\Delta$ of the task at hand, that is, for every closed simplex $\sigma\in \m{I}$, it must be the case that 
$
\delta (\Xi^t(\sigma))\subseteq \Delta(\sigma).
$
This guarantees that, for every closed input state $\sigma\in\m{I}$, which may evolve in any of the states $\tau_1,\dots,\tau_k$ of $\m{P}^{(t)}$ after $t$ rounds (i.e., after $t$ applications of $\Xi$), the output~$\delta(\tau_i)$ of each of these states is legal w.r.t. the specification $\Delta$ of the task, stating that the output state~$\delta(\tau_i)$  must be one of the states listed in~$\Delta(\sigma)$. Note that if $\sigma$ is closed, then so are all simplices in $\Xi(\sigma)$, and therefore $\Xi^t(\sigma)\neq\varnothing$. The inclusion $\delta (\Xi^t(\sigma))\subseteq \Delta(\sigma)$ is not enforced for open simplices $\sigma$. This is because there are no executions in $\m{M}$ that let the processes in $\sigma$  solely interacting among themselves.  
Finally, the map $\delta$ must be \emph{name-independent}, that is, for every two vertices $(i,w)$ and $(j,w)$ of~$\m{P}^{(t)}$, one must have $\delta(i,w)=(i,x)$ and $\delta(j,w)=(j,x)$, for the same output value~$x$.  This is because the name~$i$ of process~$i$ is not part of its input. The identifier $\id(p_i)$ may however be part of $p_i$'s input. Obviously, a setting assuming that $\id(p_i)=i$ for every $i\in[n]$, may equivalently  be viewed as assuming no IDs, and then disgarding the constraint of name-independence by allowing the outcomes of~$\delta$ to depend on the names.   

\subparagraph{Computability}

The following theorem (whose proof can be found in Appendix~\ref{app:iff-general}, for the sake of completeness) fully characterizes the ability to solve a task in a given number of rounds. It can be summarized as ``the diagram in Figure~\ref{fig:topoiff} commutes''.  Its statement is a straightforward generalization of similar statements in~\cite{Herlihy2013}. 

\begin{theorem}\label{theo:iff-general}
Let $\m{M}$ be a communication model, let $\Xi$ be the associated communication map, and let $(\m{I},\m{O},\Delta)$ be a task. For every $t\geq 0$, there exists a $t$-round algorithm solving $(\m{I},\m{O},\Delta)$ in model~$\m{M}$ if and only if there exists a chromatic name-independent simplicial map $\delta:\m{P}^{(t)}\to \m{O}$ such that, for every closed simplex $\sigma\in \m{I}$, $\delta (\Xi^t(\sigma))\subseteq \Delta(\sigma)$. 
\end{theorem}

\section{Speedup Theorem}
\label{sec:speed}

This section presents another illustration of the flexibility of our model, by establishing a result generalizing to other models the speedup theorem by Brandt~\cite{Brandt19} stated for \local. 
Given a task $\m{T}=(\m{I},\m{O},\Delta)$ Brandt's speedup theorem enables to automatically construct a task $\m{T}'=(\m{I},\m{O}',\Delta')$ such that $\m{T}$ is solvable in $t$~rounds in $\local(G)$ if and only if $\m{T}'$ is solvable in $t-1$~rounds in $\local(G)$. Although the theorem holds for specific graphs~$G$ only, and under some conditions to be satisfied by the task~$\m{T}$, it provides a powerful tool for the design of lower bounds. For extending this theorem to models beyond \local, we define the notion of speedup in general. Let  $\m{M}$ be a communication model, and let  $\Xi$ be its associated communication map. 

\begin{definition}\label{Def:speedup}
Let $\m{T}=(\m{I},\m{O},\Delta)$ be a task. A task $\m{T}'=(\m{I},\m{O}',\Delta')$ is a \emph{speedup} of $\m{T}$ for $\m{M}$ if the following two conditions hold: 
\begin{enumerate}
    \item for every $t\geq 1$, if there exists a $t$-round algorithm solving $\m{T}$ in~$\m{M}$, then there exists a simplicial map $\alpha:\m{P}^{(t-1)}\to \m{O}'$ such that, for every closed simplex $\sigma\in \m{I}$, $\alpha (\Xi^{t-1}(\sigma))\subseteq \Delta'(\sigma)$; 
    \item there exists a simplicial map $\beta:\Xi(\m{O}')\to \m{O}$ such that, for every closed simplex $\sigma\in \m{I}$, $\beta(\Xi(\Delta'(\sigma)))\subseteq \Delta(\sigma)$.
\end{enumerate}
\end{definition}

Figure~\ref{fig:speedup} provides a graphical representation of Definition~\ref{Def:speedup}. The first condition expresses the fact that the task $\m{T}'$ is solvable in $t-1$ rounds whenever the original task $\m{T}$ is solvable in $t$~rounds. The second condition essentially expresses the fact that the task $\m{T}$ is solvable in a single round whenever the processes are given as input a solution for~$\m{T}'$. The second condition therefore guarantees that the round-complexity of task $\m{T}'$ does not decrease too much compared to the complexity of task~$\m{T}$.

\begin{lemma}\label{lem:speedup}
If a task $\m{T}'$ is a speedup of a task $\m{T}$ in~$\m{M}$, then, for every $t\geq 1$, $\m{T}$ is solvable in at most $t$ rounds in~$\m{M}$ if and only if $\m{T}'$ is solvable in at most $t-1$ rounds in~$\m{M}$. 
\end{lemma}

\begin{proof}
    If $\m{T}$ is solvable in at most $t$ rounds then the first condition of Definition~\ref{Def:speedup} guarantees the existence of a simplicial map $\alpha:\m{P}^{(t-1)}\to \m{O}'$ that agrees with $\Delta'$. Thanks to  Theorem~\ref{theo:iff-general}, this guarantees that $\m{T}'$ is solvable in $t-1$ rounds. 

Conversely, let us assume that $\m{T}'$ is solvable in at most $t-1$ rounds. A $t$-round algorithm for $\m{T}$ proceeds in two phases, as follows. During the first phase, $t-1$ communication rounds are performed, and a solution for $\m{T}'$ is computed at each process. During the second phase, a single round of communications is performed, during which the processes exchange their outputs for $\m{T}'$. After this final round, every process outputs the value resulting from the application of $\beta$ to the collection of outputs in $\val(\m{O}')$ collected during the second phase. Let $\sigma\in\m{I}$ be a closed simplex. By Theorem~\ref{theo:iff-general}, the first phase results in a global state $\tau\in\Delta'(\sigma)$. Let $\tau'\in\Xi(\tau)$. The second condition of Definition~\ref{Def:speedup} guarantees that $\beta(\tau')\in \Delta(\sigma)$. Therefore, our $t$-round algorithm allows the processes to output a solution for $\m{T}$ that agrees with the input~$\sigma$. 
\end{proof}

\section{Generalized Brandt's Speedup Mechanism}
\label{sec:generalizedBrandt}

We now introduce a general approach susceptible to construct a speedup task $\m{T}'$ of any given task~$\m{T}=(\m{I},\m{O},\Delta)$, under any model~$\m{M}$. Let $\delta:\m{P}^{(t)}\to \m{O}$ be a simplicial map that agrees with~$\Delta$. There is a natural candidate $\m{T}'=(\m{I},\m{O}',\Delta')$ that may be a speedup of~$\m{T}$. Given a vertex $(i,v)\in\m{P}^{(t-1)}$, recall that we denote by $\St(i,v)$ the actual closure $\Cl(\St(i,v))$ of the star~$\St(i,v)$. Hence, $\St(i,v)$ is a complex (while $\St(i,v)$ is usually the set of all simplices of $\m{P}^{(t-1)}$ containing~$(i,v)$). 
Let us set 
$
\alpha(i,v)=\Big (i,\delta\big (\Xi(\St(i,v))\big )\Big)
$
for every $(i,v)\in\m{P}^{(t-1)}$. We have that $\Xi(\St(i,v))$ is a complex, that is pure, and of dimension~$n-1$. This complex is actually a sub-complex of~$\m{P}^{(t)}$. The simplicial map~$\delta$ maps this sub-complex to a sub-complex of~$\m{O}$. In other words, the image of a vertex $(i,v)\in\m{P}^{(t-1)}$ is a pair $(i,\m{K})$ where $\m{K}$ is a pure $(n-1)$-dimensional sub-complex of~$\m{O}$. Therefore, a good candidate for $\m{O}'$ is the chromatic simplicial complex on vertex set 
$
V(\m{O}')=\{\alpha(i,v):(i,v)\in \m{P}^{(t-1)}\}, 
$
such that a non-empty set $\{(i,\m{K}_i):i\in I\}$ of vertices of $\m{O}'$, with $I\subseteq [n]$, is a simplex of $\m{O}'$ if there exists a simplex $\{(i,v_i):i\in I\}\in \m{P}^{(t-1)}$ such that, for every $i\in I$, $\m{K}_i=\alpha(i,v_i)$. By construction, the map $\alpha:\m{P}^{(t-1)}\to \m{O}'$ is simplicial, which is the first requirement of Definition~\ref{Def:speedup}. However, it is uneasy to push this approach further, for two reasons. First, $\m{O}'$ depends on~$\delta$, and, second, it is not clear which conditions need to be satisfied for guaranteeing the existence of~$\beta:\Xi(\m{O}')\to \m{O}$. Yet, we were inspired by this approach for our generalization of Brandt's speedup theorem. In essence, the approach followed in the proof of Brandt's speedup theorem consists of finding a clever way to split the complex $\delta\big (\Xi(\St(i,v))\big )$ into a collection of sets of output values, one for each each edge incident to~$p_i$ in the graph~$G$ of $\local(G)$. We shall follow the same approach, for each channel incident to process~$i$ in~$\m{M}$. 

First, we present a speedup construction satisfying the first condition of Definition~\ref{Def:speedup} only. Obviously, satisfying that first condition is straightforward, by taking any task~$\m{T}'=(\m{I},\m{O}',\Delta')$ solvable in zero rounds (e.g., pick $\m{O}'=\Cl(\{(i,\bot):i\in[n]\})$ with $\Delta'(\sigma)=\big\{\{(i,\bot):i\in \name(\sigma)\}\big\}$ for any $\sigma\in\m{I}$). Our construction is non-trivial in the sense that, under additional assumptions, it also satisfies the second condition of Definition~\ref{Def:speedup}. In absence of these assumptions, only the first condition is guaranteed to hold, and thus we call that result the \emph{weak speedup lemma}. Note that, depending on the context, such a lemma may potentially yield a task~$\m{T}'$ whose complexity is smaller than~$t-1$  but non-necessarily zero (e.g., $t/2$, or $\sqrt{t}$), which might be sufficient to establish non-trivial lower bounds. 
Let $\m{T}=(\m{I},\m{O},\Delta)$ be a task for $n$ processes, and let $\m{M}$ be a communication model with  $\Xi$ its associated communication map. For defining a speedup task $\m{T}'=(\m{I},\m{O}',\Delta')$, 
we define the complex~$\m{O}'$, and the input-output specification~$\Delta'$, as follows. 

\noindent $\bullet$ The vertices of~$\m{O}'$ are a subset of all the pairs $(i,\bb{P}_i)$ with 
$\bb{P}_i = (x_i , \bb{S}_i)$,  
$x_i\in \val(\m{I})$, and 
\[
\bb{S}_i=\{\bb{S}_{i,e,E_i} : ((i,E_i) \in \m{M}) \land (e \subseteq [n]) \land (i \in e) \},
\]
where,  for every vertex $(i,E_i) \in \m{M}$, and for every $e \subseteq [n]$ and $i \in e$, 
$
\bb{S}_{i,e,E_i} \in 2^{2^{\Sk_i(\m{O})}}.
$
In other words, $\bb{S}_{i,e,E_i}$ is a collection of sets with elements in $\Sk_i(\m{O})$. Note that $\Sk_i(\m{O})$ is merely a set of vertices of~$\m{O}$, each of the form $(i,y)$ for some $y\in \val(\m{O})$. There is a set $\bb{S}_{i,e,E_i}$ for every set $E_i$ of channels incident to~$i$ susceptible to be active in some communication, and for every potential channel $e\subseteq [n]$. 
In fact, each set $\bb{S}_{i,e,E_i}\in \bb{S}_i$ is identified by a pair $(e,E_i)$. That is, formally, $\bb{S}_i$~is  an array indexed by pairs (\emph{channel}, \emph{set of channels}). Nevertheless, for the sake of simplifying the notations, we describe $\bb{S}_i$ as a set.
For $(i,(x_i , \bb{S}_i))$ to be a vertex of~$\m{O}'$, the sets $\bb{S}_{i,e,E_i}$ in $\bb{S}_i$ must satisfy the following property: 

\begin{description}
\item[P0:] For every vertex $(i,E_i) \in \m{M}$ with $E_i = \{e_1,\ldots,e_d\}$, and for every $(S_{i,e_1,E_i},\ldots,S_{i,e_d,E_i}) \in \bb{S}_{i,e_1,E_i} \times \ldots \times \bb{S}_{i,e_d,E_i}$, we have 
$
\bigcap\limits_{j \in [d]} S_{i,e_j,E_i} \neq \emptyset.
$
\end{description}

\medbreak 

\noindent $\bullet$ A vertex-set $\{(i,(x_i,\bb{S}_i)): i \in [n]\}$ is a facet of $\m{O}'$ if, for every closed simplex $\varphi = \{(i,E_i) : i \in I \} \in \m{M}$, and for every $e = \{i_1,\ldots,i_d\} \in E_i$ for some $i \in I$, when we denote $\bar{E_{i_k}}=\varphi \setminus (i_k,E_{i_k})$ for every $k \in e$, the following two properties hold:

\begin{description}
\item[P1:] For every $e = \{i_1,\ldots,i_d\} \in E_i$, there exists $(S_{i_1,e,E_{i_1},\bar{E_{i_1}}},\ldots,S_{i_d,e,E_{i_d},\bar{E_{i_d}}}) \in \bb{S}_{i_1,e,E_{i_1}} \times \ldots \times \bb{S}_{i_d,e,E_{i_d}}$ such that, for every $\big ((i_1,y_1),\ldots,(i_d,y_d)\big) \in S_{i_1,e,E_{i_1},\bar{E_{i_1}}} \times \ldots \times S_{i_d,e,E_{i_d},\bar{E_{i_d}}}$, there exists $\tau \in \m{I}$ for which
\[
\{(i_1,x_{i_1}),\dots, (i_d,x_{i_d})\} \subseteq \tau,
\;\mbox{and}\;
 \{ (i_1,y_1),\ldots,(i_d,y_d) \} \in \Cl(\Delta(\tau)) . 
\]

Moreover, for every $k \in [d]$,
$$
\forall \Bar{E_{i_k}} \in \Cl(\St(i_k,E_k)), \, S_{i_k,e,E_i,\bar{E_{i_k}}} = S_{i_k,e,E_i,E_{i_k}} 
$$

\item[P2:] $|\bb{S}_{i,\{i\},E_i}| = 1$ (i.e., there is a unique set in~$\bb{S}_{i,\{i\},E_i}$), and, for every $j \in I$, for every $e \in E_j$, if $i \in e$ and $j \notin \cup_{f \in E_i} f$, that is, if $p_i$ does not receive information from~$p_j$ in any channel of~$E_i$, then $\bb{S}_{i,e,E_i} = \bb{S}_{i,\{i\},E_i}$.  
\end{description}

\medbreak 

\noindent $\bullet$ Finally, the input-output specification $\Delta'$ satisfies the following:
\begin{description}
\item[P3:] For every two simplices $\sigma = \{(i,x_i) : i \in I \} \in \m{I}$, and $\tau = \{(i,(x'_i,\bb{S}_i)) : i \in I\} \in \m{O}'$, where $I \subseteq [n]$, 
\[
    \tau \in \Delta'(\sigma) \iff \forall i \in I, x'_i = x_i.
\]
\end{description}

\subsection{From $t$ Rounds to $t-1$ Rounds} 
\label{sec:t_to_t-1}

We show that, under certain conditions, the task $\m{T}'=(\m{I},\m{O}',\Delta')$ defined above is a \emph{weak} speedup of~$\m{T}$, that is it verifies the first condition of the definition \ref{Def:speedup}. The main condition is actually \emph{not} local checkability, but a \emph{local independence} property. This property is related to the protocol complex at time~$t$, and is local to every process. Given a vertex $(i,E_i) \in \m{M}$ and a channel $e\in E_i$, let us consider a simplex $\{(j,v_j):j\in e\}\in \m{P}^{(t)}$. The views $v_j$ at time~$t$ of the processes $j\in e$ are mutually consistent. Let $j\in e$, and let us recall that $\St(j,v_j)$ denotes the complex induced by the set of all simplices of $\m{P}^{(t)}$ containing vertex $(j,v_j)$. For every $j\in e$, and every channel $f \in E_j \smallsetminus \{e\}$, one can consider a simplex $\sigma_{f,j} = \{ (k,v_k) : k \in f \} \in \St(j,v_j)$. The $t$-independence property essentially states that the simplices $\sigma_{f,j}$, for $j\in e$ and $f\in E_j \smallsetminus \{e\}$ are ``independent'' of each other, in the sense that, if, for every $j\in e$ and every $f\in E_j \smallsetminus \{e\}$, the views of the processes~$k\in f$ are consistent with the view of process~$j\in e$, then the views of all these processes are mutually consistent all together. That is, the union of all the simplices $\sigma_{f,j}$ is a simplex of~$\m{P}^{(t)}$. 

\begin{definition}\label{def:t-independence}
Let $t\geq 1$ be an integer. $\m{T}$ satisfies the \emph{$t$-independence} property w.r.t.~$\m{M}$ if, for every closed simplex $\{(i,E_i) : i \in I \} \in \m{M}$, for every $i \in I$, for every $e \in E_i$, for every simplex ${\{(j,v_j):j\in e\}\in \m{P}^{(t)}}$, and for every collection 
\[
{\big\{\sigma_{f,j} = \{ (k,v_k) : k \in f \} \in \St(j,v_j) : (j\in e) \wedge (f \in E_j \smallsetminus \{e\})\big\}},
\]
we have 
\begin{equation}\label{eq:independence}
\bigcup_{j\in e} \; \bigcup_{f \in E_j \smallsetminus \{e\}} \sigma_{f,j} \in \m{P}^{(t)}. 
\end{equation}
\end{definition}

Note that the $t$-independence property depends solely on the model~$\m{M}$, and on the input complex~$\m{I}$ of the task. Indeed, $\m{I}$ and~$\m{M}$ are the only parameters that govern the properties of the protocol complex at time~$t$. An interesting special case of the $t$-independence property is when considering the ``self-loop'' $e=\{i\}$ (recall that we assumed that $\{i\}\in E_i$ for every $(i,E_i) \in \m{M}$). For the self-loops, the $t$-independence property implies that, for every $(i,E_i) \in \m{M}$, for every vertex $(i,v_i)\in \m{P}^{(t)}$, and for every collection $\big\{\sigma_{f} = \{ (j,v_j) : j \in f \} \in \St(i,v_i) : f \in E_i \setminus \{i \} \big\}$, we have 
\begin{equation}\label{eq:independence-degenerate}
\bigcup_{f \in E_i \smallsetminus \{i \} } \sigma_{f} \in \m{P}^{(t)}.
\end{equation}

\subparagraph{Example}

Let us consider $\local(G)$ where $G$ is the graph on $n=2k$ nodes $u_1,\dots,u_k,v_1,\dots,v_k$, with edges $\{u_1,v_i\}$ and $\{u_i,v_1\}$ for $i=1,\dots,k$. In other words, $G$ consists of two stars centered at $u_1$ and~$v_1$, respectively, and sharing the edge~$\{u_1,v_1\}$. In~$G$,  for $i\in [k]$, process~$2i-1$ occupies node $u_i$, and process~$2i$ occupies node $v_i$. Let $\m{T}=(\m{I},\m{O},\Delta)$ be a task where $\m{I}$ is the input complex corresponding to $3$-coloring (without IDs). Specifically, $\rho=\{(i,x_i):i \in [n]\}$ is a facet of~$\m{I}$ if, for every $i\in [n]$, $x_i\in\{1,2,3\}$,  and, for every $i\in [k]$, $x_1\neq x_{2i}$ and $x_2\neq x_{2i-1}$. 

We illustrate 0-independence by demonstrating that, for $t=0$,   Eq.~\eqref{eq:independence} holds for $p_1$, and for the channel $e$ between $p_1$ and~$p_2$ (the extension to other processes, and to other channels is straightforward). As a warm up, we start by showing that Eq.~\eqref{eq:independence-degenerate} holds for~$p_1$. Process~$1$ is occupying node~$u_1$, and we have 
\[
E_1=\big \{\{1\},\{1,2\},\{1,4\},\dots,\{1,2k\}\big \}.
\]
Let $f\in E_1\smallsetminus \{1\}$, say $f=\{1,2i\}$ for some $i\in [k]$, and let us assume that $p_1$ is colored~$x_1\in\{1,2,3\}$. The set  $\sigma_f=\{(1,x_1),(2i,x_{2i})\}$ is a simplex of $\m{P}^{(0)}=\m{I}$ in $\St(1,x_1)$ whenever $x_{2i}\in\{1,2,3\}\smallsetminus\{x_1\}$. We have 
\[
\bigcup_{f \in E_1 \smallsetminus \{\{1\}\} } \sigma_{f} = 
\{(1,x_1),(2,x_{2}),(4,x_4),\dots,(2k,x_{2k})\},
\]
and indeed this simplex is in~$\m{I}$ as $x_{2i}\neq x_1$ for every $i\in[k]$. Thus Eq.~\eqref{eq:independence-degenerate} is satisfied by~$p_1$. Let us now turn our attention to the channel $e=\{1,2\}\in E_1$, and let us assume that  $p_1$ is colored~$x_1\in\{1,2,3\}$ while $p_2$ is colored~$x_2\in\{1,2,3\}\smallsetminus\{x_1\}$. Let $j\in e=\{1,2\}$, and $f\in E_j\smallsetminus \{e\}$. For $j=1$, $f=\{1,2i\}$ for some $i\in [2,k]$, and the set $\sigma_{f,1}=\{(1,x_1),(2i,x_{2i})\}$ is a simplex of $\m{P}^{(0)}=\m{I}$ in $\St(1,x_1)$ whenever $x_{2i}\in\{1,2,3\}\smallsetminus\{x_1\}$. Similarly, for $j=2$, $f=\{2,2i-1\}$ for some $i\in [2,k]$, and the set $\sigma_{f,2}=\{(2,x_2),(2i-1,x_{2i-1}\}$ is a simplex of $\m{P}^{(0)}=\m{I}$ in $\St(2,x_2)$ whenever $x_{2i-1}\in\{1,2,3\}\smallsetminus\{x_2\}$. We have 
\[
\bigcup_{j\in \{1,2\}} \; \bigcup_{f \in E_j \smallsetminus \{\{1,2\}\}} \sigma_{f,j} = \{(1,x_1),(2,x_{2}), (3,x_3),(4,x_4), (5,x_5), \dots,(2k,x_{2k})\}. 
\]
This simplex is actually a facet of~$\m{I}$, and thus is in $\m{P}^{(0)}=\m{I}$. Thus Eq.~\eqref{eq:independence} is satisfied by~$p_1$ and the channel $e=\{1,2\}\in E_1$. The same argument can be used to show that Eq.~\eqref{eq:independence} is satisfied for every process~$i\in [n]$, for any channel~$e\in E_i$, which demonstrates 0-independence. 

This example can be extended to the graph $G$ consisting of two full $(k-1)$-ary trees of height~${h>1}$ (each node has $k-1$ children, excepted the leaves, which are all at distance~$h$ from the root), connected by an edge $e=\{i,j\}$ connecting process~$i$ to process~$j$. For instance, for $t\leq h-1$, and for a vertex $(i,v_i)$ of~$\m{P}^{(t)}$, the view~$v_i$ is a tree rooted at~$p_i$ spanning all nodes of $G$ at distance at most~$t$ from~$p_i$, whose nodes are properly 3-colored. A set $\{(i,v_i),(j,v_j)\}$ is a simplex of~$\m{P}^{(t)}$ whenever the two views $v_i$ and $v_j$ are compatible. That is, all processes in $v_i\cap v_j$ are colored the same in the two views $v_i$ and $v_j$, and the colors of processes $p_i$ and $p_j$ are different. In this context, the $t$-independence property essentially says that whatever color is given to each node not in $v_i\cup v_j$ but adjacent to a leaf of $v_i$ or $v_j$, if this color is different from the color of the leaf it is attached to, the resulting set 
\[
\{(i,v_i),(j,v_j)\} \cup \{(k,v_k):k\in (N_G(i)\cup N_G(j)) \smallsetminus \{i,j\}\},
\]
where, for every $k\in (N_G(i)\cup N_G(j)) \smallsetminus \{i,j\}$, $v_k$ is the view at distance~$t$ of a neighbor of $p_i$ or $p_j$ that is compatible with $v_i$ and $v_j$, with the nodes not in $v_i\cup v_j$ colored arbitrarily as above, is indeed a simplex of~$\m{P}^{(t)}$. 

\medbreak

The following lemma does not require local checkability, but solely the local independence property. 

\begin{lemma}\label{weakspeedup}
Let $\m{T}=(\m{I},\m{O},\Delta)$ be a task for $n$ processes, and let $\m{M}$ be a model on~$n$ vertices. Let $t\geq 1$ be an integer, and assume that $\m{T}$ satisfies the $(t-1)$-independence property w.r.t.~$\m{M}$. If $\m{T}$ is solvable in $t$ rounds in~$\m{M}$, then the task $\m{T}'=(\m{I},\m{O}',\Delta')$ is solvable in $t-1$ rounds in~$\m{M}$, where $\m{O}'$ is the complex defined by properties~P0-2, and $\Delta'$ is the input-output relation defined by~P3.
\end{lemma}

\begin{proof}
Let $\delta : \m{P}^{(t)} \to \m{O}$ solving $\m{T}$ in $t$ rounds in~$\m{M}$. We define $\alpha:\m{P}^{(t-1)} \to \m{O}'$, and then show that $\alpha$ is solving $\m{T}'$ in~$\m{M}$. Intuitively the algorithm corresponding to~$\alpha$ consists in simulating every possible solution for $\m{T}$ when the values in~$\m{P}^{(t-1)}$ of some subset of processes are fixed, and when the communications occurring at time~$t$ are fixed. 
For a closed simplex $\varphi\in\m{M}$, and a chromatic complex~$\m{K}$, we define  
\[
\Xi(\m{K},\varphi)=\Cl\Big(\bigcup_{\sigma\in\m{K}:\name(\sigma)=\name(\varphi)}\Xi(\sigma,\varphi)\Big).
\]
Formally, note that, for any two  closed simplices $\varphi$ and $\psi$ in $\St(i,E_i)$, $\Sk_i (\Xi (\St(\sigma),\varphi)) = \Sk_i (\Xi (\St(\sigma),\psi))$ for every simplex~$\sigma\in\m{P}^{(t-1)}$. Therefore, we can abuse notation by denoting $\Sk_i (\Xi (\St(\sigma),\varphi))$ as $\Sk_i (\Xi (\St(\sigma),E_i))$. For any vertex $(i,v_i) \in \m{P}^{(t-1)}$, we let 
\[
\alpha(i,v_i)= (i,\bb{P}_i)  \;\mbox{with} \; \bb{P}_i = (x_i,\bb{S}_i),
\]
where $x_i$ is the input value of process~$i$ (which is present in its view~$v_i$), and 
\[
\bb{S}_i = \{\bb{S}_{i,e,E_i} : ((i,E_i) \in \m{M}) \land (e \subseteq [n]) \land (i \in e) \},
\]
where 
\[
\ora{\bb{S}_{i,e,E_i} =  \{ \delta(\Sk_i( \Xi (\St(\sigma),E_i)))   : (\sigma \in \Sk_e(\St(i,v_i))) \land (\dim(\sigma) = |e| - 1) \}.}
\]
For every $\sigma \in \Sk_e(\St(i,v_i))$ with $\dim(\sigma) = |e| - 1$, the set 
\[
S_{i,e,E_i}^\sigma  =  \delta(\Sk_i (\Xi (\St(\sigma),E_i)))
\]
is the set of every possible output for process~$i$ using $\delta$ whenever the communication pattern~$E_i$ occurred at time~$t$, and the processes in~$e$ have their value fixed according to $\sigma\in\m{P}^{(t-1)}$. In particular, we have $S_{i,e,E_i}^\sigma \in 2^{\Sk_i(\m{O})}$, and thus $\bb{S}_{i,e,E_i} \in 2^{2^{\Sk_i(\m{O})}}$, as desired. We show that P0 holds. For every vertex $(i,E_i) \in \m{M}$, let us consider a set 
\[
\{S_{i,e,E_i} \in \bb{S}_{i,e,E_i} : e \in E_i \}.
\]
By definition, for every set $S_{i,e,E_i}$, there exists $\sigma_e \in \St(i,v_i)$ such that 
\[
S_{i,e,E_i} = S_{i,e,E_i}^{\sigma_e} = \delta(\Sk_i(\Xi(\St(\sigma_e),E_i))). 
\]
Using the $(t-1)$-independence property for $e=\{i\} \in E_i$, it holds that,
\[
\bigcup_{e \in E_i \setminus \{ i \} } \sigma_e \in \m{P}^{(t-1)}.
\]
There is only one candidate for the simplex  corresponding to the channel $e=\{i\}$, and this simplex is $\sigma_{\{i\}}=\{(i,v_i)\}$. Therefore,
\[
\bigcup_{e \in E_i \setminus \{ i \} } \sigma_e = \bigcup_{e \in E_i} \sigma_e \in \m{P}^{(t-1)}.
\]
For every $e \in E_i$, we have 
\[
\Sk_i(\Xi(\St(\cup_{e \in E_i} \sigma_e),E_i)) \subseteq \Sk_i(\Xi(\St(\sigma_e),E_i)).
\]
Therefore,
\[
\delta(\Sk_i(\Xi(\St(\cup_{e \in E_i} \sigma_e),E_i))) \subseteq \bigcap_{e \in E_i} S_{i,e,E_i}. 
\]
Now, $\delta(\Sk_i(\Xi(\St(\bigcup_{e \in E_i} \sigma_e),E_i)))$ cannot be empty, simply because $\bigcup_{e \in E_i} \sigma_e \in \m{P}^{(t-1)}$. Therefore, property P0 holds, that is, $\alpha$ produces vertices of~$\m{O}'$.

\medbreak 

To prove that $\alpha$ solves $\m{T}'$, it is sufficient to consider an arbitrary facet $\rho = \{(i,v_i) : i \in [n] \} \in \m{P}^{(t-1)}$, and its image $\alpha(\rho)=\{ (i,(x_i,\bb{S}_i) : i \in [n] \}$, and we show that $\alpha(\rho)$ is a facet of~$\m{O}'$ that agrees with~$\Delta'$. It is sufficient to show that both properties~P1 and~P2 hold as, by definition of $\alpha$, P3~holds by construction. 
 
 \medbreak 

First we prove that P1 holds. For every closed simplex $\varphi = \{(i,E_i) : i \in I \} \in \m{M}$, and for every $e = \{i_1,\ldots,i_d \} \in E_i$ for some $i \in I$, we consider the face $\sigma = \{(i_k,v_k) : k \in [d]\}$ of~$\rho$, which is indeed of dimension~$|e| - 1$. Note that, since $(i_k,v_k) \in \sigma$ for every $k \in [d]$, we have $\sigma \in \bigcap_{k = 1}^{d} \Sk_e(\St(i_k,v_k))$. For every $k \in [d]$, let us consider the set 
\[
S_{i_k,e,E_{i_k}}^\sigma = \delta(\Sk_{i_k}(\Xi(\St(\sigma),E_{i_k}))) \in \bb{S}_{i_k,e, E_{i_k}}. 
\]

Note that $S_{i_k,e,E_{i_k}}^\sigma$ is by definition independent of $E_{i_\ell}$ for $\ell \in [d] \setminus \{k\}$ therefore it is enough to show that $S_{i_k,e,E_{i_k}}^\sigma$ satisfies the first part of~P1 and let $((i_1,y_1),\ldots,(i_d,y_d)) \in S_{i_1,e,E_{i_1}}^\sigma \times \ldots \times S_{i_d,e,E_{i_d}}^\sigma$. Let $k \in [d]$. By definition of $\alpha$,  there exists $\sigma_k \in \St (\sigma)$ with $\name(\sigma_k) = (\cup_{f\in E_{i_k}} f)\smallsetminus e$ such that $(i_k,y_k) \in \delta(\Sk_{i_k}(\Xi(\sigma_k \cup \sigma,E_{i_k})))$. In fact, if we fix the communication pattern~$E_{i_k}$ occurring at time~$t$, and if we fix the values in $\m{P}^{(t-1)}$ of all the processes in the channels of $E_{i_k}$, there exists exactly one vertex $(i_k,w_k) \in \m{P}^{(t)}$ that is consistent with the fixed communication pattern, and with the fixed values in $\m{P}^{(t-1)}$. Therefore the fixed values yield exactly one output for process~$i_k$. In other words, for every $k \in [d]$, 
\[
\delta(\Sk_{i_k}(\Xi(\sigma_k \cup \sigma,E_{i_k}))) = \{ (i_k,y_k) \}.
\]
Using the $(t-1)$-independence property, it holds that $(\bigcup_{k = 1}^{d} \sigma_k) \cup \sigma \in \m{P}^{(t-1)}$. It follows that, for every $k \in [d]$, we have
\[
\emptyset \neq \delta(\Sk_{i_k}(\Xi((\cup_{l = 1}^{d}\sigma_l)\cup\sigma,E_{i_k}))) \subseteq \delta(\Sk_{i_k}(\Xi(\sigma_k\cup\sigma,E_{i_k}))) = \{ (i_k,y_k) \}
\]
As a consequence, 
\[
\delta(\Sk_{\{i_1,\ldots,i_d \}}(\Xi((\cup_{l = 1}^{d}\sigma_l)\cup\sigma,\varphi)))= \{ (i_k,y_k) : k \in [d] \}.  
\]
Recall that, by definition of $\alpha$, for any $j \in [d]$, $\alpha(i_j,v_{i_j})= (i_j,(x_{i_j},\bb{S}_{i_j}))$, 
where $x_{i_j} \in \val(\m{I})$ is the input of process $i_j$ in~$v_{i_j}$. Let us consider a facet $\phi$  of $\m{I}$ such that $(\cup_{k = 1}^{d}\sigma_k)\cup\sigma \in \Delta(\phi)$. Note that, in particular, $\{(i_j,x_{i_j}) : j \in [d] \} \subseteq \phi$. Moreover, since $\delta$ solves $\m{T}$,  we also have $\{ (i_k,y_k) : k \in [d] \} \in \Cl(\Delta(\phi))$. It follows that Property~P1 holds.

\medbreak

Second, we show that P2 holds. For every closed simplex $\varphi = \{(i,E_i) : i \in I\} \in \m{M}$, for every $i\in I$ and for every $e \in E_j$, if $i \in e$ and $j \notin \cup_{f \in E_i} f$, we show that 
\[
\bb{S}_{i,e,E_i} = \bb{S}_{i,\{i\},E_i} \triangleq \{ \delta(\Sk_i(\Xi(\St(i,v_i),E_i)))\}.
\]
Note that $|\bb{S}_{i,\{i\},E_i}| = 1$.
Let $S\in \bb{S}_{i,e,E_i}$. We show that $S = \delta(\Sk_i(\Xi(\St(i,v_i),E_i)))$. By definition of $\bb{S}_{i,e,E_i}$, there exists a simplex $\sigma_0 \in \Sk_e(\St(i,v_i))$ such that $S = \delta(\Sk_i( \Xi (\St(\sigma_0),E_i)))$. Using the $(t-1)$-independence property on~$e$,  it follows that, for every $\sigma \in \Sk_e(\St(i,v_i))$, and for every collection $\{ \sigma_f \in \Sk_f(\St(i,v_i)) : f \in E_i \smallsetminus e \}$ of simplices,  
\[
    \Big( \bigcup_{f \in E_i \setminus e} \sigma_f \Big) \cup \sigma \in \m{P}^{(t-1)}.
\]
Now, $j \in e$ but, for every $f \in E_i$, $j \notin f$. It follows  that $e \notin E_i$, and thus $E_i \smallsetminus e  = E_i$.  Therefore $\cup_{f \in E_i \smallsetminus e} \sigma_f$ defines the values of all the processes~$k$ that are in at least one channel of~$E_i$. This entirely characterizes the output  of the process~$i$ by~$\delta$. Therefore, we have established that $\delta(\Sk_i( \Xi (\St(\sigma_0),E_i)))$ is actually independent of the simplex $\sigma_0\in\Sk_e(\St(i,v_i))$. Formally, 
$$ 
\delta(\Sk_i( \Xi (\St(\sigma_0),E_i))) = \delta(\Sk_i( \Xi (\St(i,v_i),E_i))),
$$ 
which implies that P2 holds, and concludes the proof. 
\end{proof}

\subsection{From $t-1$ Rounds to $t$ Rounds} 
\label{sec:t-1_to_t}

It follows from Lemma~\ref{weakspeedup} that the weak version of Brandt's speedup Theorem can be extended from \local\/ to asynchronous computing models with crashes. The strong version of the speedup theorem does not only require the local independence property, but also the \emph{edge-checkability} property, a stronger variant of local decidability (cf. Def.~\ref{def:local-decidability}).

Granted with the framework of Section~\ref{sec:general-framework}, we define a general notion of \emph{local decidability} (also referred to as \emph{local checkability}). 
We generalize the specific notion defined for \local\/ (see~\cite{FraigniaudKP13}), and the specific notion defined for \wf\/ (see~\cite{FraigniaudRT13}). Given a simplex $\{(i,v_i):i\in I\}$ of a state complex~$\m{K}$, and given $J\subseteq I$, let $\pi_J(\sigma)=\{(i,v_i):i\in J\}$. Note that $\pi_J(\sigma)$ is different from $\Sk_J(\sigma)$ as the former is a simplex, while the latter is a complex. 

\begin{definition}\label{def:local-decidability}
A task $(\m{I},\m{O},\Delta)$ is \emph{locally checkable} for the communication model $\m{M}$ if, for every $\sigma\in\m{I}$ with $\name(\sigma) = I$, for every set $\tau = \{(i,y_i) : i\in I \land y_i \in \val(\m{O})\}$, and for every closed simplex $\varphi=\{(i,E_i):i\in I\}$ of $\m{M}$, the following holds: 
\[
\tau\in\Delta(\sigma) \iff \forall i\in I, \; \pi_{J_i}(\tau)\in \Delta (\pi_{J_i}(\sigma)), \; \mbox{where} \; J_i=\cup_{e\in E_i}e.
\]
The class of locally checkable tasks for the model $\m{M}$ is denoted by $\ld(\m{M})$. 
A problem $\Pi$ is  \emph{locally checkable} if, for every $(\m{T},\m{M})\in \Pi$, the task $\m{T}\in\ld(\m{M})$. 
\end{definition}

In other words, a task $(\m{I},\m{O},\Delta)$ is in~$\ld(\m{M})$ if the correctness of a potential solution $\tau=\{(i,y_i): i\in I\}\in\m{O}$ for an input $\sigma=\{(i,x_i): i\in I\}\in\m{I}$ can be checked in a single round of communication under~$\m{M}$. Indeed, assuming every process~$i\in I$ is given a pair $(x_i,y_i)$ of input-output values, any round of communication performed according to some communication pattern $\varphi\in\m{M}$ allows every process~$i\in I$ to acquire a set $\{(x_j,y_j): j\in J_i\}$ where $x_j$ and $y_j$ are the input and output values of process~$j\in J_i$ in~$\sigma$ and $\tau$, respectively. Every process~$i\in I$ can thus check whether $\pi_{J_i}(\tau)=\{(j,y_j):j\in J_i\}$ belongs to $\Delta (\pi_{J_i}(\sigma))$ or not. Locally decidability states that the output $\tau$ is correct for the input $\sigma$, where $\name(\sigma)=\name(\tau)=I\subseteq[n]$, if and only if all the individual tests performed by the processes in $I$ are passed.
In $\local$, many standard graph problems, e.g., vertex coloring and maximal independent sets (MIS), are locally checkable. Similarly, $(d+1)$-coloring  is in $\ld(\wflocal(G))$ for every $G\in \m{G}_d$. This is because the model guarantees that, for every edge $\{i,j\}$, at least one of the two processes~$i$ and~$j$ receives the color of the other process. On the other hand, as opposed to the case of $\local(G)$, MIS is not in $\ld(\wflocal(G))$. Indeed, a node that is not in the set may not receive information from all its neighbors, and thus it cannot systematically check whether  it has at least one neighbor in the set. The generalized maximal independent set (GMIS) task~\cite{KuhnZ2018} is however in $\ld(\hlocal(H))$ for every hypergraph~$H$ (in GMIS, each hyperedge~$e$ is associated with a threshold $t_e\in\{1,|e|-1\}$, and a set $S$ of vertices is a solution to GMIS if, for every $e\in E(H)$, $|e\cap S| \leq t_e$, and $S$ is maximal for this property).

\begin{definition}\label{def:edge-checkability}
A task $(\m{I},\m{O},\Delta)$ is \emph{edge-checkable} in a model $\m{M}$ if, for every $\sigma\in\m{I}$ with $\name(\sigma) = I$, for every set $\tau = \{(i,y_i) : i\in I \land y_i \in \val(\m{O})\}$, and for every closed simplex $\varphi=\{(i,E_i):i\in I\}$ of $\m{M}$, the following holds: 
\[
\tau\in\Delta(\sigma) \iff \forall i\in I, \; \forall e \in E_i, \; \pi_{e}(\tau)\in \Delta (\pi_{e}(\sigma)).
\]
\end{definition}

Note that MIS (in its standard form) is locally checkable in \local. However, MIS is not edge-checkable in \local\/ because, by considering each of its neighbors independently, a process that is not in the MIS cannot determine whether it has at least one neighboring process that is in the MIS. Nevertheless we describe further a systematic way to transform a locally checkable task into an edge-checkable task, which applies to \local, as well as to other models. As in~\cite{Brandt19}, our speedup theorem assumes an underlying mechanism enabling the processes in a same channel~$e$ to \emph{break symmetry}, whether it be thanks to a local or global identification mechanism, or thanks to an implicit or  explicit ordering of the processes in the same channel. 

\begin{theorem}\label{strspeedup}
Let $\m{T}=(\m{I},\m{O},\Delta)$ be a task for $n$ processes, and let $\m{M}$ be a model on~$n$ vertices that supports symmetry breaking in its channels. Let $t\geq 1$ be an integer. Let us assume that $\m{T}$ satisfies the $(t-1)$-independence property w.r.t.~$\m{M}$, and that $\m{T}$ is edge-checkable in~$\m{M}$. $\m{T}$~is solvable in $t$ rounds in~$\m{M}$ if and only if $\m{T}'=(\m{I},\m{O}',\Delta')$ is solvable in $t-1$ rounds in~$\m{M}$, where $\m{O}'$ is the complex defined by properties~P0-2, and $\Delta'$ is the map defined by Property~P3.
\end{theorem}

\begin{proof}
    Thanks to Lemma~\ref{weakspeedup} it is  sufficient to show the existence of a simplicial map~$\beta:\Xi(\m{O}')\to \m{O}$ such that, for every closed simplex $\sigma\in \m{I}$, $\beta(\Xi(\Delta'(\sigma)))\subseteq \Delta(\sigma)$. 
    Let $\tau = \{{(i,(x_i,\bb{S}_i)) : i \in I} \}\in\Delta'(\sigma)$ for some closed simplex $\sigma = \{(i,x_i) : i \in I \} \in \m{I}$, note that $\tau$ is a face of a facet of $\m{O}'$, which, by definition, satisfy P1 and P2. Since $\m{T}$ is edge-checkable, it is sufficient to prove that, for every closed simplex $\varphi=\{(i,E_i):i\in I\} \in \m{M}$, for every $i\in I$, process~$i$ can output a solution $(i,y_i) \in \Sk_i(\m{O})$ such that,  for every $e\in E_i$, $\{(j,y_j) : j \in e \} \in \Delta(\pi_e(\sigma))$. After one round of communication according to~$\varphi$, every  process~$i\in I$ receives messages from every process $j \in e$ for every $e\in E_i$. Therefore every process~$i$ has access to the set $\{ (x_k,\bb{S}_k) : k \in J_i \}$ where $J_i=\cup_{e\in E_i}e$. More specifically, for every $e \in E_i$,  process~$i$ has access to $\{ (x_k,\bb{S}_k) : k \in e \}$. If $j\in e$ does not receive from~$i$, property~P2 guarantees that $|\bb{S}_{j,e,E_j}| = 1$. For the other processes $j\in e$, the symmetry-breaking mechanism and property~P1 allows these processes to choose sets $S_{j,e,E_j} \in \bb{S}_{j,e,E_j}$  such that, for every choice of $\{y_j: j \in e\}$ where $(j,y_j)\in S_{j,e,E_j}$ for every $j\in e$, there exists $\tau \in \m{I}$ such that
    \[
    \{(j,x_j) : j \in e\} \subseteq \tau, \; \mbox{and} \;
    \{(i,y_j) : j \in e \} \in \Delta(\tau). 
    \]
    Since $\m{T}$ is edge-checkable, $\{(j,y_j) : j \in e \} \in \Delta(\tau)$ implies that 
    \[
    \{(j,y_j) : j \in e \} \in \Delta(\pi_e(\tau)) = \Delta(\pi_e(\sigma)).
    \]
    By repeating this operation for every channel $e \in E_i$, process~$i$ can output any value $y_i$ such that,
    \[
        (i,y_i) \in \bigcap_{e \in E_i} S_{i,e,E_i}. 
    \]
    Such a value $y_i$ does exist thanks to property~P0. The correctness of this algorithm is straightforward since the task is edge-checkable,  and the set $S_{i,e,E_i}$ are precisely chosen to satisfy~P1.
\end{proof}

\subsection{Applications}
\label{sec:application2hypergraphs}

As mentioned before, several models satisfy the conditions in the statement of Theorem~\ref{strspeedup}, beyond \local. This is for instance the case of dynamic graph models, $\dyn(\m{F})$, under some conditions on~$\m{F}$. This is also the case of hypergraph models, $\hlocal(H)$, under some conditions on~$H$. This is even the case for asynchronous crash-prone models such as \wf \, with 2 processes, this has already been detailed in Section \ref{sec:appwf}. Therefore we focus here on giving a complete proof of the application of Theorem \ref{strspeedup} to \hlocal. 

First, we show that, in \hlocal, edge-checkability  is actually essentially equivalent, up to $\pm 1$ rounds, to  local decidability (see proof in Appendix~\ref{app:lem:locdecvsedgdec}). 

\begin{lemma}\label{lem:locdecvsedgdec}
Let $H$ be a hypergraph on $n$ nodes. For every task $\m{T}=(\m{I},\m{O},\Delta)$ on $n$ processes, if $\m{T}$ is locally checkable in $\hlocal(H)$, then there exists a task $\m{T}'=(\m{I},\m{O}',\Delta')$ that is edge-checkable in $\hlocal(H)$, such that (1)~any solution for~$\m{T}$ can be transformed into a solution for~$\m{T}'$ via a single round of communication in $\hlocal(H)$, and (2)~any solution for~$\m{T}'$ can be transformed into a solution for~$\m{T}$ in zero rounds.
\end{lemma}

The following is a direct consequence of Theorem~\ref{strspeedup}. 

\begin{corollary}\label{hypergraphspeedup}
Let  $H$ be a hypergraph on~$n$ nodes, let $\m{T}=(\m{I},\m{O},\Delta)$ be a task for $n$ processes, and let $t\geq 1$. Let us assume that $\m{T}$ satisfies the $r$-independence property w.r.t.~$\hlocal(H)$ for every $r\in\{0,\dots,t-1\}$, and that $\m{T}$ is edge-checkable in~$\hlocal(H)$. Let us  assume the existence of a symmetry breaking mechanism among the processes in each hyperedge of~$H$.  
$\m{T}$ is solvable in $t$ rounds in~$\hlocal(H)$ if and only if $\m{T}^{(t)}=(\m{I},\m{O}^{(t)},\Delta^{(t)})$ is solvable in zero rounds in~$\hlocal(H)$, where $\m{T}^{(t)}$ is the task defined by iterating $t$~times the construction defined by properties~P0-3.
\end{corollary}

\begin{proof}
    Thanks to Theorem~\ref{strspeedup}, it is sufficient to prove that the task $\m{T}'=(\m{I},\m{O}',\Delta')$ obtained by applying the construction defined by properties~P0-3 still satisfies the precondition in the statement of the theorem. First, $\m{T}'$ is edge-checkable. Indeed, the conditions specified by P0-3 are locally checkable, as the communication pattern is unique in $\hlocal(H)$, and Property~P2 is trivially satisfied in $\hlocal(H)$ as the model in undirected (all processes in a hyperedge $e$ play the same role, as far as the channel~$e$ is concerned). $\m{T}'$ satisfies the $(t-2)$-independence property w.r.t $\hlocal(H)$, simply because $\m{T}$ itself satisfies the $(t-2)$-independence property, and $\m{T}$ and $\m{T}'$ have the same input complexes~$\m{I}$. 
\end{proof}

For instance, let us consider $\hlocal(H)$ where  $H$ is a \emph{linear} hypergraph~\cite{KuhnZ2018}, that is, for every two distinct hyperedges $e$ and $f$, $|e \cap f| \leq 1$. Let us also assume that $H$ is locally isomorphic to a (linear) \emph{hypertree}. Formally, $H$ has no small cycle, that is, its girth is greater than $2t+1$. Linear hypertrees can be viewed as graphs where hyperedges are displayed as cliques (cf. Fig.~\ref{fig:hypertree} for an example). Note that any lower-bound for $\hlocal(H)$ holds for $\local(G)$ where $G$ is obtained from $H$ by replacing the hyperedges of~$H$ by cliques in~$G$. 
To show that the $r$-independence property holds for any $r \geq 0$ in such a hypergraph, let us fix an hyperedge $e \in E(H)$, and, for $i \in e$ and $f \in E_i \smallsetminus e$, let $N_r(i,f)$ denote the sub-hypertree induced by all the nodes of~$H$ at distance at most $r$ from node~$i$ via the edge~$f$. The sub-hypertrees $N_r(i,f)$ for $i\in e$ and $f \in E_i \smallsetminus e$ are vertex-disjoint, and thus the $r$-independence property is satisfied in absence of labeling, or with an edge-checkable labeling.   
The class of linear hypergraphs with large girth contains the graphs that are locally isomorphic to a regular tree, which is precisely the class of graphs on which Brandt's speedup theorem~\cite{Brandt19} has been originally proven.  However, linear hypergraphs form a much larger class, which includes graphs with short cycles, like the one depicted on Fig.~\ref{fig:hypertree}. 

\section{Conclusion}

In this paper, we have extended Brandt's speedup theorem from \local\/ to  round-based based communication models supporting full-information protocols, including many standard synchronous communication models in networks, and even asynchronous models such as \wflocal\/, and \wf\/ for 2 processes. 
Extending the speedup theorem to \wf\/ for more than 2 processes remains open. Our approach consisted in decomposing every subcomplex $\delta(\Xi(\St(i,v_i)))$ of the output complex, where $(i,v_i)$ is a vertex of~$\m{P}^{(t-1)}$, into collections of sets $\bb{S}_{i,e,E_i}$, one for each potential channel~$e$, and for each possible local communication pattern~$E_i$ of the underlying communication model~$\m{M}$. This approach is well suited to models like \local, \hlocal, or even \dyn\/ and \wf\/ for two processes, but it does not match the characteristics of $\wf$ for larger systems, essentially because \wf\/ does not satisfy the local independence property whenever $n>2$. Nevertheless, we believe that there is another way to decompose the complexes $\delta(\Xi(\St(i,v_i)))$, for all $(i,v_i)\in\m{P}^{(t-1)}$, that would provide a speedup theorem for model not satisfying local independence (e.g., \wf), but this decomposition still remains to be found.

\bibliographystyle{plainurl}

\begin{thebibliography}{10}

\bibitem{AlcantaraCFR19}
Manuel Alcantara, Armando Casta{\~{n}}eda, David Flores{-}Pe{\~{n}}aloza, and
  Sergio Rajsbaum.
\newblock The topology of look-compute-move robot wait-free algorithms with
  hard termination.
\newblock {\em Distributed Comput.}, 32(3):235--255, 2019.

\bibitem{AttiyaCHP19}
Hagit Attiya, Armando Casta{\~{n}}eda, Maurice Herlihy, and Ami Paz.
\newblock Bounds on the step and namespace complexity of renaming.
\newblock {\em {SIAM} J. Comput.}, 48(1):1--32, 2019.

\bibitem{AttiyaW04}
Hagit Attiya and Jennifer Welch.
\newblock {\em Distributed Computing: Fundamentals, Simulations, and Advanced
  Topics}.
\newblock Series on Parallel and Distributed Computing. Wiley, 2004.

\bibitem{Balliu0CORS19}
Alkida Balliu, Sebastian Brandt, Yi{-}Jun Chang, Dennis Olivetti, Mika{\"{e}}l
  Rabie, and Jukka Suomela.
\newblock The distributed complexity of locally checkable problems on paths is
  decidable.
\newblock In {\em 38th {ACM} Symposium on Principles of Distributed Computing
  (PODC)}, pages 262--271, 2019.

\bibitem{Balliu0EHMOS20}
Alkida Balliu, Sebastian Brandt, Yuval Efron, Juho Hirvonen, Yannic Maus,
  Dennis Olivetti, and Jukka Suomela.
\newblock Classification of distributed binary labeling problems.
\newblock In {\em 34th International Symposium on Distributed Computing
  (DISC)}, volume 179 of {\em LIPIcs}, pages 17:1--17:17. Schloss Dagstuhl -
  Leibniz-Zentrum f{\"{u}}r Informatik, 2020.

\bibitem{Balliu0HORS19}
Alkida Balliu, Sebastian Brandt, Juho Hirvonen, Dennis Olivetti, Mika{\"{e}}l
  Rabie, and Jukka Suomela.
\newblock Lower bounds for maximal matchings and maximal independent sets.
\newblock In {\em 60th {IEEE} Annual Symposium on Foundations of Computer
  Science (FOCS)}, pages 481--497, 2019.

\bibitem{Balliu0OS20}
Alkida Balliu, Sebastian Brandt, Dennis Olivetti, and Jukka Suomela.
\newblock How much does randomness help with locally checkable problems?
\newblock In {\em 39th {ACM} Symposium on Principles of Distributed Computing
  (PODC)}, pages 299--308, 2020.

\bibitem{Brandt19}
Sebastian Brandt.
\newblock An automatic speedup theorem for distributed problems.
\newblock In {\em 38th {ACM} Symposium on Principles of Distributed Computing
  (PODC)}, pages 379--388, 2019.

\bibitem{CastanedaFPRRT21}
Armando Casta{\~{n}}eda, Pierre Fraigniaud, Ami Paz, Sergio Rajsbaum, Matthieu
  Roy, and Corentin Travers.
\newblock A topological perspective on distributed network algorithms.
\newblock {\em Theor. Comput. Sci.}, 849:121--137, 2021.

\bibitem{CastanedaR10}
Armando Casta{\~{n}}eda and Sergio Rajsbaum.
\newblock New combinatorial topology bounds for renaming: the lower bound.
\newblock {\em Distributed Comput.}, 22(5-6):287--301, 2010.

\bibitem{CastanedaR12}
Armando Casta{\~{n}}eda and Sergio Rajsbaum.
\newblock New combinatorial topology bounds for renaming: The upper bound.
\newblock {\em J. {ACM}}, 59(1):3:1--3:49, 2012.

\bibitem{CasteigtsFQS11}
Arnaud Casteigts, Paola Flocchini, Walter Quattrociocchi, and Nicola Santoro.
\newblock Time-varying graphs and dynamic networks.
\newblock In {\em 10th International Conference on Ad-hoc, Mobile, and Wireless
  Networks (ADHOC-NOW)}, volume 6811 of {\em Lecture Notes in Computer
  Science}, pages 346--359. Springer, 2011.

\bibitem{Chang20}
Yi{-}Jun Chang.
\newblock The complexity landscape of distributed locally checkable problems on
  trees.
\newblock In {\em 34th International Symposium on Distributed Computing
  (DISC)}, volume 179 of {\em LIPIcs}, pages 18:1--18:17. Schloss Dagstuhl -
  Leibniz-Zentrum f{\"{u}}r Informatik, 2020.

\bibitem{ChangKP19}
Yi{-}Jun Chang, Tsvi Kopelowitz, and Seth Pettie.
\newblock An exponential separation between randomized and deterministic
  complexity in the {LOCAL} model.
\newblock {\em {SIAM} J. Comput.}, 48(1):122--143, 2019.

\bibitem{FajstrupGHMR16}
Lisbeth Fajstrup, Eric Goubault, Emmanuel Haucourt, Samuel Mimram, and Martin
  Raussen.
\newblock {\em Directed Algebraic Topology and Concurrency}.
\newblock Springer, 2016.

\bibitem{Ferreira04}
Afonso Ferreira.
\newblock Building a reference combinatorial model for manets.
\newblock {\em {IEEE} Netw.}, 18(5):24--29, 2004.

\bibitem{FischerGK17}
Manuela Fischer, Mohsen Ghaffari, and Fabian Kuhn.
\newblock Deterministic distributed edge-coloring via hypergraph maximal
  matching.
\newblock In {\em 58th {IEEE} Annual Symposium on Foundations of Computer
  Science (FOCS)}, pages 180--191. {IEEE} Computer Society, 2017.

\bibitem{FischerLP85}
Michael~J. Fischer, Nancy~A. Lynch, and Mike Paterson.
\newblock Impossibility of distributed consensus with one faulty process.
\newblock {\em J. {ACM}}, 32(2):374--382, 1985.

\bibitem{FraigniaudKP13}
Pierre Fraigniaud, Amos Korman, and David Peleg.
\newblock Towards a complexity theory for local distributed computing.
\newblock {\em J. {ACM}}, 60(5):35:1--35:26, 2013.

\bibitem{FraigniaudP20}
Pierre Fraigniaud and Ami Paz.
\newblock The topology of local computing in networks.
\newblock In {\em 47th International Colloquium on Automata, Languages, and
  Programming (ICALP)}, volume 168 of {\em LIPIcs}, pages 128:1--128:18.
  Schloss Dagstuhl - Leibniz-Zentrum f{\"{u}}r Informatik, 2020.

\bibitem{FraigniaudRT13}
Pierre Fraigniaud, Sergio Rajsbaum, and Corentin Travers.
\newblock Locality and checkability in wait-free computing.
\newblock {\em Distributed Comput.}, 26(4):223--242, 2013.

\bibitem{FraigniaudRT20}
Pierre Fraigniaud, Sergio Rajsbaum, and Corentin Travers.
\newblock A lower bound on the number of opinions needed for fault-tolerant
  decentralized run-time monitoring.
\newblock {\em J. Appl. Comput. Topol.}, 4(1):141--179, 2020.

\bibitem{GodardP16}
Emmanuel Godard and Eloi Perdereau.
\newblock k-set agreement in communication networks with omission faults.
\newblock In {\em 20th International Conference on Principles of Distributed
  Systems (OPODIS)}, volume~70 of {\em LIPIcs}, pages 8:1--8:17. Schloss
  Dagstuhl - Leibniz-Zentrum f{\"{u}}r Informatik, 2016.

\bibitem{GoosS16}
Mika G{\"{o}}{\"{o}}s and Jukka Suomela.
\newblock Locally checkable proofs in distributed computing.
\newblock {\em Theory Comput.}, 12(1):1--33, 2016.

\bibitem{GoubaultMT18}
{\'{E}}ric Goubault, Samuel Mimram, and Christine Tasson.
\newblock Geometric and combinatorial views on asynchronous computability.
\newblock {\em Distributed Comput.}, 31(4):289--316, 2018.

\bibitem{GrahamRS90}
Ronald~L. Graham, Bruce~L. Rothschild, and Joel~H. Spencer.
\newblock {\em Ramsey theory}, volume~20.
\newblock John Wiley and Sons, 1990.

\bibitem{Herlihy2013}
Maurice Herlihy, Dmitry~N. Kozlov, and Sergio Rajsbaum.
\newblock {\em Distributed Computing Through Combinatorial Topology}.
\newblock Morgan Kaufmann, 2013.

\bibitem{HerlihyS99}
Maurice Herlihy and Nir Shavit.
\newblock The topological structure of asynchronous computability.
\newblock {\em J. {ACM}}, 46(6):858--923, 1999.

\bibitem{HirvonenS2020}
Juho Hirvonen and Jukka Suomela.
\newblock Distributed algorithms.
\newblock Aalto University, Finland, 2020.

\bibitem{KormanKP10}
Amos Korman, Shay Kutten, and David Peleg.
\newblock Proof labeling schemes.
\newblock {\em Distributed Comput.}, 22(4):215--233, 2010.

\bibitem{KuhnZ2018}
Fabian Kuhn and Chaodong Zheng.
\newblock Efficient distributed computation of {MIS} and generalized {MIS} in
  linear hypergraphs.
\newblock {\em CoRR}, abs/1805.03357, 2018.

\bibitem{KuttenNP014}
Shay Kutten, Danupon Nanongkai, Gopal Pandurangan, and Peter Robinson.
\newblock Distributed symmetry breaking in hypergraphs.
\newblock In {\em 28th International Symposium on Distributed Computing
  (DISC)}, volume 8784 of {\em Lecture Notes in Computer Science}, pages
  469--483. Springer, 2014.

\bibitem{Linial92}
Nathan Linial.
\newblock Locality in distributed graph algorithms.
\newblock {\em {SIAM} J. Comput.}, 21(1):193--201, 1992.

\bibitem{MendesTH14}
Hammurabi Mendes, Christine Tasson, and Maurice Herlihy.
\newblock Distributed computability in byzantine asynchronous systems.
\newblock In {\em 46th ACM Symposium on Theory of Computing (STOC)}, pages
  704--713, 2014.

\bibitem{NaorS95}
Moni Naor and Larry~J. Stockmeyer.
\newblock What can be computed locally?
\newblock {\em {SIAM} J. Comput.}, 24(6):1259--1277, 1995.

\bibitem{Peleg00}
David Peleg.
\newblock {\em Distributed Computing: A Locality-Sensitive Approach}.
\newblock Discrete Mathematics and Applications. SIAM, 2000.

\bibitem{RozhonG20}
V{\'{a}}clav Rozhon and Mohsen Ghaffari.
\newblock Polylogarithmic-time deterministic network decomposition and
  distributed derandomization.
\newblock In {\em 52nd {ACM} Symposium on Theory of Computing (STOC)}, pages
  350--363, 2020.

\bibitem{SaksZ00}
Michael~E. Saks and Fotios Zaharoglou.
\newblock Wait-free k-set agreement is impossible: The topology of public
  knowledge.
\newblock {\em {SIAM} J. Comput.}, 29(5):1449--1483, 2000.

\bibitem{Suomela20}
Jukka Suomela.
\newblock Using round elimination to understand locality.
\newblock {\em {SIGACT} News}, 51(3):63--81, 2020.

\end{thebibliography}

\bigbreak
\appendix
\centerline{\bf \Large A P P E N D I X}

\section{The Topology of Distributed Computing}
\label{sec:distributed-tasks}

We consider a system of $n\geq 1$ independent autonomous processes labeled from $1$ to~$n$, exchanging information via some communication medium. It is convenient to express our framework using the language of combinatorial  topology, which enables to place various distributed models under the same umbrella of terminologies and concepts. We follow the general approach in~\cite{Herlihy2013}, and this section recalls the main characteristics  of this approach, including the central definition of distributed \emph{tasks}. 

\subsection{Elements of Topology}
\label{app:element-topo}

A  \emph{simplicial complex} is defined by a vertex set~$V$, and a collection $\m{K}$ of non-empty subsets of~$V$, closed by inclusion. That is, if $\sigma\in \m{K}$, then every non-empty set $\sigma'\subseteq \sigma$ belongs to~$\m{K}$. Every set in $\m{K}$ is called a \emph{simplex}. The set~$V$ is often clear from the context, in which case one merely refers to a complex by the collection~$\m{K}$, and to the vertex set of $\m{K}$ as $V(\m{K})$. 

The dimension of a simplex $\sigma$ is $|\sigma|-1$, hence a vertex is a simplex of dimension zero. A \emph{face} of a simplex~$\sigma$ is any simplex $\sigma'\subseteq \sigma$. A \emph{facet} is a simplex that is maximal, i.e., not included in any other simplices. Note that a complex can be described by the list of its facets. A complex is \emph{pure} is all its facets have the same dimension. The dimension of a pure complex is the dimension of any of its facets. 

A sub-complex of a complex~$\m{K}$ is a subset of $\m{K}$ that is a complex. The \emph{star} of a simplex $\sigma$ in a complex~$\m{K}$, denoted by $\St(\sigma)$, is the set of simplices of $\m{K}$ having $\sigma$ as a face. The star of $\sigma$ naturally induces a sub-complex of~$\m{K}$, composed of all simplices of~$\m{K}$ included in at least one simplex of~$\St(\sigma)$. This complex is merely the \emph{closure}, $\Cl(\St(\sigma))$, of the star of $\sigma$ in~$\m{K}$, where the closure of a set $S$ of simplices of~$\m{K}$, denoted by $\Cl(S)$, is the smallest simplicial subcomplex of $\m{K}$ that contains each simplex in~$S$. In fact, in this paper, we will abuse notation, and will refer to $\St(\sigma)$ as a complex, that is, $\St(\sigma)$  actually refers to~$\Cl(\St(\sigma))$. 

All complexes considered in this paper are \emph{chromatic}. That is, each of their vertices has the form $(i,x)$, where $i\in [n]=\{1,\dots,n\}$ is a process index (the ``color'' of the vertex), and $x$ is a value that depends on the context, and no simplices can contain two vertices with the same color. To avoid confusion when considering problems such as graph coloring, we refer to the index~$i\in[n]$ of a process as its \emph{name} (and not its color), and the non-empty index set~$I\subseteq [n]$ of a simplex $\sigma=\{(i,x_i):i\in I \}$ is denoted by $\name(\sigma)$. Given a non-empty set $I\subseteq [n]$, $\Sk_I(\m{K})$ denotes the \emph{skeleton} subcomplex of~$\m{K}$ composed of all simplices $\sigma\in\m{K}$ with $\name(\sigma)\subseteq I$.

A particular class of mappings  between simplicial complexes plays a crucial role in the topological framework applied to distributed computing: those preserving simplices.  Specifically, given two complexes $\m{K}$ and $\m{K}'$, a map $f:V(\m{K})\to V(\m{K}')$ is \emph{simplicial} if, for every $\sigma\in \m{K}$, $f(\sigma)\in\m{K}'$, where $f(\sigma)=\{f(v):v\in \sigma\}$. All maps considered in this paper apply to chromatic complexes, and are name-preserving, i.e., $f(i,x)=(i,y)$ for every $(i,x)\in V(\m{K})$. We say that such maps are chromatic. 

Given a finite set $X$ of values, the $(n-1)$-dimensional chromatic \emph{pseudo-sphere} induced by~$X$ is the complex $\m{S}_n(X)$ whose vertices are all pairs $(i,x)$ with $i\in[n]$ and $x\in X$, and every non-empty set $\{(i,x_i):i\in I\}$ of vertices, where $I\subseteq [n]$, forms a simplex. In short, 
\[
\m{S}_n(X)=\Big \{ \big \{(i,x_i):i\in I \big\} : (\varnothing\neq I\subseteq [n]) \wedge (\forall i\in I, x_i\in X)\Big\}.
\]

\subsection{Distributed States}

This section recalls how all possible states of a distributed system at a given time can be captured by a single combinatorial object, namely the \emph{state complex}. Specific instantiations of state complexes are the input complexes, the output complexes, and the protocol complexes, described further in the text. 

 At any point in time, all possible global states of a distributed system can be represented as an $(n-1)$-dimensional complex. The vertices of this complex are of the form $(i,s)$ where $i\in [n]$ is the name of a process, $s\in X$ is a state of process~$i$, and $X$ is the set of all possible states of a process. A non-empty set $\{(i,s_i):i\in I\}$ of such vertices, $I\subseteq [n]$, forms a simplex if the states~$s_i$, $i\in I$, are mutually compatible. Mutual compatibility is a notion which depends on the context (e.g., input complex or output complex), and on the communication model (e.g., protocol complexes), but it should soon appear clear further in the paper. In general, a \emph{state} complex of an $n$-process system with local states in $X$ is a pure $(n-1)$-dimensional sub-complex of the chromatic pseudo-sphere $\m{S}_n(X)$. 
 
 The fact that a vertex $(i,s)$ belongs to two different simplices $\sigma$ and $\sigma'$ with $\name(\sigma)=\name(\sigma')$ means that process~$i$ in local state~$s$ cannot distinguish the global state~$\sigma$ from the global state~$\sigma'$.  More specifically, even if  process~$i$ is aware that the system is in global state~$\sigma$ or $\sigma'$, it remains uncertain about the local state of every other process~$j$ such that $(j,s')\in \sigma$ and $(j,s'')\in \sigma'$ with~$s'\neq s''$. 

\subparagraph{Example.}

A system maintaining a clock $c_i\in \mathbb{F}_q$ at every process $i\in[n]$, with bounded drift~$d$ between processes, i.e., $(c_i-c_j) \bmod q \leq d$ for every $i,j\in[n]$, has a state complex $\m{K}\subseteq \m{S}_n(\mathbb{F}_q)$ where, for every non-empty $I\subseteq [n]$:
$
\{(i,c_i):i\in I\} \in \m{K} \iff \forall i,j\in I, \; (c_i-c_j) \bmod q \leq d.
$

\subsection{Input and Output Complexes, and Distributed Tasks}
\label{subsec:distributed-tasks}

This section recalls the important notion of \emph{task}, which formalizes the typical ``functions'' to be computed in the distributed setting. A task is defined by three objects: the input complex, the output complex, and an input-output specification.   

\subparagraph{Input Complex.} 

Let  $\mathbb{I}$ be a finite set, whose elements are called input values. An input complex~$\m{I}$ with values in~$\mathbb{I}$  is a pure $(n-1)$-dimensional  sub-complex of the $(n-1)$-dimensional chromatic pseudo-sphere $\m{S}_n(\mathbb{I})$. 
For instance, for binary consensus, $\mathbb{I}=\{0,1\}$, and all sets $\big \{(i,v_i): i\in I \big\}$ with $I\in[n]$ and  $v_i\in \{0,1\}$ for every $i\in I$, are simplices of~$\m{I}$, i.e., $\m{I}=\m{S}_n(\{0,1\})$. Note that it is often the case that $\m{I}=\m{S}_n(\mathbb{I})$, as in consensus. However, other problems  assumes $\m{I}\subset \m{S}_n(\mathbb{I})$. A typical example is $(m,k)$-renaming, for $m\geq k\geq n$, in which the $n$ processes are given as input $n$ distinct integers in the set~$[m]$, and must output $k$ distinct integers in~$[k]$. In this case a non-empty set $\{(i,v_i):i\in I\}$ with $I\subseteq [n]$ is a simplex of $\m{I}$ if $v_i\in [m]$ for every $i\in I$, and $v_i\neq v_j$ for every distinct $i,j\in I$. 

\subparagraph{Output Complex.} 

An output complex $\m{O}$ is a pure $(n-1)$-dimensional sub-complex of the $(n-1)$-dimensional chromatic pseudo-sphere $\m{S}_n(\mathbb{O})$  induced by  a finite set $\mathbb{O}$, whose elements are called output values.
For instance, in the case of binary consensus, $\mathbb{O}=\{0,1\}$, and a non-empty set $\big \{(i,v_i): i\in I \big\}$ with $I\subseteq [n]$,  is a simplex of $\m{O}$ if $v_i\in \mathbb{O}$ for every $i\in I$, and $v_i=v_j$ for every two $i,j\in I$. Note that the output complex has generally more structure than the input complex. 

\subparagraph{Input-Output Relation.} 

The input-output relation specifies, for every input state, the collection of  output states that are legal w.r.t.~this input.  Specifically, the input-output specification is a function $\Delta:\m{I}\to 2^{\m{O}}$ that is returning, for every input simplex $\sigma\in\m{I}$, a non-empty collection $\tau_1,\dots,\tau_k$ of output simplices, $k\geq 1$, such that $\name(\tau_i)=\name(\sigma)$ for every $i=1,\dots,k$. As $\Delta$ is name-preserving, it is called \emph{chromatic}.  
For instance, in the case of binary consensus, $\Delta$~maps every input simplex $\sigma=\{(i,v_i):i\in I\}\in\m{I}$ to the output simplices $\tau_0=\{(i,0):i\in I\}$ and $\tau_1=\{(i,1):i\in I\}$ whenever there are two processes $i,j\in I$ with $v_i\neq v_j$ in~$\sigma$. Instead, for every $x\in\{0,1\}$, $\Delta$~maps the simplex $\sigma=\{(i,x):i\in I\}\in \m{I}$ to the simplex $\tau_x=\{(i,x):i\in I\}\in \m{O}$. 

\bigbreak

We now have all the ingredients to define what is a \emph{task}. Note that this definition is independent of the communication model.

\begin{definition}\label{def:task}
A \emph{task} in a $n$-process system is a triple $(\m{I},\m{O},\Delta)$, where $\m{I}$ is the input complex, $\m{O}$ is the output complex, both of dimension $n-1$, and $\Delta:\m{I}\to 2^{\m{O}}$ is the input-output specification. 
\end{definition}

\subsection{Examples}

For the readers more familiar with distributed graph problems (e.g., coloring, MIS, etc.) than with distributed system problems (e.g., consensus, renaming, etc.), let us define the task $(\m{I},\m{O},\Delta)$ of $k$-coloring the vertices of an $n$-node cycle~$C_n$, with identifiers (IDs) in $[N]=\{1,\dots,N\}$, $N\geq n$.  A vertex of the input complex~$\m{I}$ is a pair of the form $(i,(x,\{y,z\}))$, with $i\in [n]$, $\{x,y,z\}\subseteq [N]$, and $x\neq y\neq z\neq x$. The semantics of such a vertex is that process~$i$ is handling some node of $C_n$ which received $x$ as ID, and the  neighbors of this node in~$C_n$ received IDs $y$ and $z$.  A non-empty set $\{(i,(x_i,\{y_i,z_i\})):i\in I\}$ with $I\subseteq [n]$ is a simplex of~$\m{I}$ if the nodes of~$C_n$ can be assigned distinct IDs in~$[N]$ such that,  for every $i\in I$, a node~$u$ has ID~$x_i$, and $\{y_i,z_i\}$ is the set of IDs of the two neighbors of~$u$.  
A vertex of the output complex~$\m{O}$ is a pair of the form $(i,(x,\{y,z\},c))$, with $c\in [k]=\{1,\dots,k\}$. Any non-empty set $\{(i,(x_i,\{y_i,z_i\},c_i)):i\in I\}$ of vertices forms a simplex of the output complex if the same condition as for the simplices of~$\m{I}$ holds, and, in addition, for every $i\neq j$ in~$I$, $c_i\neq c_j$ whenever $x_j\in \{y_i,z_i\}$ or $x_i\in \{y_j,z_j\}$. Let $\sigma=\{(i,(x_i,\{y_i,z_i\})):i\in I\}$ in~$\m{I}$. 
A simplex $\tau=\{(i,(x'_i,\{y'_i,z'_i\},c_i)):i\in I\}$ of~$\m{O}$ satisfies $\tau\in \Delta(\sigma)$ if, for every $i\in I$, $(x'_i,\{y'_i,z'_i\}=(x_i,\{y_i,z_i\})$. 

Some models, e.g., the \local\/ model~\cite{Linial92,Peleg00}, implicitly encode a network~$G$ in the model itself. Every process~$i$ is actually viewed as located at a node~$i$ of~$G$, which is exchanging information with the neighbors of node~$i$ in~$G$ only. In this case, the relevant task $(\m{I}',\m{O}',\Delta')$  is to $k$-color the vertices of the specific graph~$G$ itself, with the additional constraints that the matching between the processes and the nodes of~$G$ is fixed (but not known to the processes). Therefore, the vertices of~$\m{I}'$ are merely defined as pairs $(i,x)$ with $x\in[N]$, and a non-empty set $\{(i,x_i):i\in I\}$ with $I\subseteq [n]$ is a simplex of~$\m{I}'$ if $x_i\neq x_j$ for every two indices $i,j\in I$. The output complex~$\m{O}'$ has vertices $(i,c)$ with $c\in[k]$, and a non-empty set $\{(i,c_i):i\in I\}$ with $I\subseteq [n]$ is a simplex of~$\m{O}'$ if, for every $i,j\in I$, $c_j\neq c_i$ whenever $j\in N_G(i)$. For every $\sigma\in \m{I}$, $\Delta'(\sigma)=\{\tau\in \m{O}': \name(\tau)=\name(\sigma)\}$.

These two variants of $k$-coloring, i.e., the tasks $(\m{I},\m{O},\Delta)$  and $(\m{I}',\m{O}',\Delta')$ above, are actually independent of the communication model. In particular, one can aim at solving the task $(\m{I},\m{O},\Delta)$ for a graph~$G$ but in the \local\/ model with network~$H$, even if $G$ and $H$ are different. Conversely,  one can aim at solving the task $(\m{I}',\m{O}',\Delta')$ in a communication model different from the \local\/ model with network~$G$. 

Finally, observe that, for the task  $(\m{I}',\m{O}',\Delta')$ to be non-trivial in the \local\/ model with network~$G$, it is required that the processes do not use their names for choosing their colors, as otherwise the processes could simply agree in advance on a specific $k$-coloring of $G$. For instance, if $G$ is the $n$-node path where every node~$i=2,\dots,n-1$ is adjacent to nodes $i+1$ and $i-1$, even $2$-coloring is trivial if processes can use their name: process~$i$ merely outputs $i\bmod 2$. To be relevant, the algorithm must therefore be \emph{name-independent}. Such a constraint appears in various distributed computing settings, e.g., when aiming at solving the renaming task. We came back to the delicate issue \emph{names vs. identifiers} in the section \ref{sec:general-framework}. 

\section{Proof of Theorem~\ref{theo:iff-general}}
\label{app:iff-general}

Let us assume that there exists a $t$-round algorithm $\alg$ solving $(\m{I},\m{O},\Delta)$ in model~$\m{M}$. This algorithm produces an output value $\alg(s)$ for every possible local state~$s$ of any process~$i$ after $t$~rounds. By definition of the local state~$s$,  the pair $(i,s)$ is a vertex of $\m{P}^{(t)}$. Let us define the map $\delta:\m{P}^{(t)}\to \m{O}$ as
\[
\delta(i,s)\triangleq (i,\alg(s)).
\]
By definition, $\delta$ is chromatic, and name-independent. Moreover, $\alg$ guarantees that, starting from any legal global input state, i.e., any simplex~$\sigma\in \m{I}$, a legal global output state is produced, i.e., a simplex~$\tau\in\m{O}$ is produced. The simplex is precisely the set $\tau={\{(i,\alg(s_i)):i\in \name(\sigma)\}}$ where $s_i$ is the local state of process~$i\in\name(\sigma)$ after $t$~rounds. It follows that ${\{\delta(i,s_i) :i\in \name(\sigma)\}\in\m{O}}$, and therefore $\delta$ is simplicial. Let us now consider a closed simplex $\sigma\in\m{I}$, and an execution of~$\alg$ in which all processes of~$\sigma$ communicate solely among themselves. By definition of~$\Xi$, after $t$~rounds, every process~$i\in \name(\sigma)$ ends up in a state~$s_i$ such that ${\{(i,s_i):i\in\name(\sigma)\}}$ is a simplex of~$\Xi^t(\sigma)$. $\alg$ then outputs $\alg(s_i)$ at every process~$i$. Since $\alg$ is correct, the resulting output simplex ${\tau=\{(i,\alg(s_i)):i\in\name(\sigma)\}}$ forms a global output state which is legal w.r.t.~the global input state~$\sigma$. In other words, $\tau\in \Delta(\sigma)$, from which it follows that $\delta (\Xi^t(\sigma))\subseteq \Delta(\sigma)$, as desired. 

Conversely, let us assume that there exists a chromatic name-independent simplicial map $\delta:\m{P}^{(t)}\to \m{O}$ such that, for every closed $\sigma\in \m{I}$, $\delta (\Xi^t(\sigma))\subseteq \Delta(\sigma)$. Let us define the $t$-round algorithm~$\alg$ as
\[
\alg(s)\triangleq \val \big (\delta(i,s)\big), 
\]
where $\val(i,x)=x$ for every vertex $(i,x)$ of any chromatic complex ($x$ is the value of that vertex). Since $\delta$ is name-independent, $\alg$ is well defined. Moreover, since $\delta$ is simplicial, the image $\tau=\{(i,\alg(s_i)):i\in \name(\sigma)\}$ of a global state $\{(i,s_i):i\in \name(\sigma)\}$ of the system after $t$~rounds starting from a global input state~$\sigma$ is a legal global output state. Moreover, since $\delta (\Xi^t(\sigma))\subseteq \Delta(\sigma)$, if the processes in $\name(\sigma)$ have communicated solely among themselves, then $\tau\in\Delta(\sigma)$, that is, $\tau$ is a legal output state w.r.t.~the input state~$\sigma$, as desired.
\qed

\section{Proof of Lemma~\ref{lem:locdecvsedgdec}}
\label{app:lem:locdecvsedgdec}

Let us define the task $\m{T}'=(\m{I},\m{O}',\Delta')$. Note that, for every $i\in [n]$, the channels incoming to~$i$ are all the hyperedges $e\in E_i$, where $E_i=E_H(i)$ is the set of hyperedges of~$H$ containing~$i$. For every $i\in[n]$, a pair $(i,K_i)$ is a vertex of~$\m{O}'$ if $K_i$ is a node-labeled sub-hypergraph of~$H$ induced by the processes in~$J_i=\cup_{e\in E_i}e$ such that, if $\ell_{i,j} = (x_{i,j},y_{i,j}) \in \val(\m{I}) \times \val(\m{O})$ denotes the label of process~$j\in J_i$, then 
\[
\{(j,y_{i,j}):j\in J_i\} \in \Delta(\{(j,x_{i,j}) : j \in J_i \}). 
\]
Let $\sigma = \{(i,x'_i):i\in I\}$ be a closed simplex of $\m{I}$. A set $\{(i,K_i):i\in I\}$ of vertices of~$\m{O}'$ is a simplex of~$\Delta'(\sigma)$ if (1)~$x_{i,i} = x'_i$ for every $i \in I$, and (2)~for every $e\in E(H)$, and for every two processes $i,j\in I\cap e$, $\pi_e(K_i)=\pi_e(K_j)$.  

We first show that $(\m{I},\m{O}',\Delta')$ is edge-checkable. Let $\sigma = \{(i,x'_i):i\in I\}$ be a closed simplex of~$\m{I}$, and let $\tau = \{(i,K_i) : i \in I\}$. 
For one direction, let us assume that $\tau  \in \Delta'(\sigma)$. We show that, for every $e\in E(H)$, $\pi_{e}(\tau)\in \Delta' (\pi_{e}(\sigma))$. 
Let $e\in E(H)$. By definition of $\Delta'$ we have that, for every $e' \in E(H)$, for every $i,j \in e \cup e'$, $\pi_e(K_i) = \pi_e(K_j)$ and $x_{i,i} = x'_i$. It follows that $\pi_e(\tau) \in \Delta'(\pi_e(\sigma))$, as desired. 
For the reciprocal, let us assume that $\pi_{e}(\tau)\in \Delta'(\pi_e(\sigma))$ for every $e\in E(H)$. 
By definition of $\Delta'$, we have that, for every $i,j \in e$, $\pi_e(K_i)=\pi_e(K_j)$, and $x_{i,i} = x'_i$. This being true for all hyperedges in~$H$,  $\tau \in \Delta'(\sigma)$ holds, and thus $\m{T}'$ is indeed edge-checkable.

Second, we show how to construct a solution for~$\m{T}'$ given any solution for~$\m{T}$, in a single round. 
Let $\tau$ be a closed simplex of $\m{O}'$ in $\Delta(\sigma)$ for some closed simplex $\sigma \in \m{I}$. In one round, every process~$i$ can collect the values and inputs of all the processes~$j\in J_i=\cup_{e\in E_i}e$. By construction, the collection of resulting labeled sub-hypergraphs of~$H$ induced by the processes in $J_i$, for $i \in \name(\sigma)$, forms a valid solution in $\Delta'(\sigma)$. 

Finally, we show how to construct a solution for~$\m{T}$ given any solution for~$\m{T}'$, in zero rounds. Let $\sigma = \{(i,x'_i):i\in I\} \in \m{I}$ be a closed simplex, and let $\tau' = \{(i,K_i):i\in I\} \in \Delta'(\sigma)$. 
We show that $\tau = \{(i,y_{i,i}):i\in I\} \in \Delta(\sigma)$. Since $(\m{I},\m{O},\Delta)$ is locally checkable, it is sufficient to prove that $\pi_{J_i}(\tau) \in \Delta(\pi_{J_i}(\sigma))$ for every $i \in I$, with $J_i=\cup_{e\in E_i}e$. Let $i\in[n]$. 
We have $\pi_{J_i}(\tau') \in \Cl(\Delta'(\sigma))$ merely because $\pi_{J_i}(\tau')$ is a face of $\tau'$. It follows that $\{\pi_e(K_j) : j \in e \} = \{\pi_e(K_i)\}$ for every $e \in E_i$. Therefore,  for every process $j\in J_i$, $x_{i,j} = x_{j,j} = x'_j$ and $y_{i,j} = y_{j,j}$. In particular, $y_{i,j} = y_{j,j}$ implies that 
\[
\pi_{J_i}(\tau) = \{(j,y_{i,j}) : j \in J_i\}. 
\]
By definition of the vertices in $\m{O}'$, we have $\{(j,y_{i,j}) : j \in J_i \} \in \Delta(\{(j,x_{i,j}) : j \in J_i\}$. Moreover $x_{i,j} = x'_j$. 
Therefore, $\pi_{J_i}(\tau) \in \Delta(\pi_{J_i}(\sigma))$ as desired, which concludes the proof.

\end{document}